\documentclass[amsart,11pt,oneside,english]{article}
\usepackage{amssymb,amsmath,amsfonts, amsmath,mathtools,hyperref,
amsthm,euscript,mathrsfs
,verbatim
,enumerate,multirow
,bbding,slashed
,multicol,color,array, 
esint,babel, tikz,tikz-cd,tikz-3dplot,tkz-graph,pgfplots,ytableau,graphicx,float,rotating,hyperref,geometry,mathdots,savesym,cite, wasysym,amscd,graphicx,pifont,float,setspace,multirow,wrapfig,picture,subfigure, amsthm,hepth, enumitem, enumerate
}
\usepackage[utf8]{inputenc}
\usepackage{thm-restate}
\usepackage{thmtools}
\usepackage{prettyref}
\usepackage[T1]{fontenc}
\usepackage{datetime}
\numberwithin{equation}{section}
\numberwithin{figure}{section}
\usetikzlibrary{arrows,positioning,decorations.pathmorphing,   decorations.markings, matrix, patterns,shapes}
\hypersetup{colorlinks=true}
\hypersetup{linkcolor=black}
\hypersetup{citecolor=black}
\hypersetup{urlcolor=black}
\theoremstyle{plain}
\newtheorem*{thm*}{Theorem}
\theoremstyle{plain}
\newtheorem{thm}{Theorem}[section]

\newtheorem{lem}[thm]{Lemma}
\newtheorem{prop}[thm]{Proposition}

\theoremstyle{definition}
\newtheorem{defn}[thm]{Definition}
\newtheorem*{defn*}{Definition}
	\newtheorem{exmp}[thm]{Example}

\newtheorem{rem}[thm]{Remark}

\usepackage[normalem]{ulem}

\tikzset{
  big arrow/.style={
    decoration={markings,mark=at position 1 with {\arrow[scale=1.5,#1]{>}}},
    postaction={decorate},
    shorten >=0.4pt},
  big arrow/.default=black}

\numberwithin{equation}{section}

\begin{document}
\begin{titlepage}
\begin{center}
\vspace{4cm}
{\Huge\bfseries   The Geometry of F$_4$-Models\\  }
\vspace{2cm}
{%
\LARGE  Mboyo Esole$^{\spadesuit}$, Patrick Jefferson$^\clubsuit$, Monica Jinwoo Kang$^\clubsuit$\\}
\vspace{1cm}

{\large $^{\spadesuit}$ Department of Mathematics, Northeastern University}\par
{\large  360 Huttington Avenue, Boston, MA 02115, USA}\par

\vspace{.3cm}
{\large $^\clubsuit$ Department of Physics, Jefferson Physical Laboratory, Harvard University}\par
{ 17 Oxford Street, Cambridge, MA 02138, U.S.A}\par
 \scalebox{.95}{\tt  j.esole@northeastern.edu,  patrickjefferson@fas.harvard.edu, jkang@fas.harvard.edu }\par
\vspace{3cm}
{ \bf{Abstract:}}\\
\end{center}
{\date{\today\  \currenttime}}

We study the geometry of elliptic fibrations satisfying the conditions of Step 8 of Tate's algorithm. We call such geometries  F$_4$-models, as the dual graph of their special fiber is the twisted affine Dynkin diagram $\widetilde{\text{F}}_4^t$. 
These geometries are used in string theory to model gauge theories with the exceptional Lie group F$_4$ on a smooth divisor $S$ of the base. 
Starting with a singular Weierstrass model of an F$_4$-model,  we present a crepant resolution of its singularities. 
We study the fiber structure of this smooth elliptic fibration and identify the fibral divisors up to isomorphism as schemes over $S$.  
These are $\mathbb{P}^1$-bundles over $S$ or double covers of $\mathbb{P}^1$-bundles over $S$. 
We compute basic topological invariants such as the double and triple intersection numbers of the fibral divisors and the Euler characteristic of the F$_4$-model. 
 In the case of Calabi-Yau threefolds, we compute  the linear form induced by the second Chern class and 
the Hodge numbers.  
We also  explore the meaning of these geometries for the physics of gauge theories in five and six-dimensional minimal supergravity theories with eight supercharges. 
 We also introduce the notion of\    \    ``frozen representations'' and explore the role of the Stein factorization in the study of fibral divisors of elliptic fibrations. 

\vfill 


\end{titlepage}

\tableofcontents

\section{Introduction}

 An elliptic fibration is a proper projective morphism $\varphi: Y\longrightarrow B$ between normal varieties such that the generic fiber is a nonsingular projective curve of genus one and the fibration is endowed with a rational section. 
Under mild assumptions, an elliptic fibration is birational to a potentially singular Weierstrass model \cite{Formulaire,MumfordSuominen}. 
The locus of points of $B$ over which the elliptic fiber is singular is called the discriminant locus. 
The discriminant locus of a Weierstrass model is a Cartier divisor that we denote by $\Delta$. 
The type of the fiber over the generic point of an irreducible component of the discriminant locus 
of an elliptic fibration is well understood following the work of Kodaira \cite{Kodaira}, N\'eron\cite{Neron}, and Tate\cite{Tate}. 
Fibers over higher dimensional loci are not classified and are the subject of much interdisciplinary research by both mathematicians and physicists \cite{Miranda.smooth,Szydlo.Thesis,Morrison:2011mb,EY,ESY1,ESY2,Braun:2013cb,Anderson:2016ler,Marsano,Hayashi:2014kca,Kuntzler:2012bu, Lawrie:2012gg,Fullwood:2012kj,Cvetic:2012xn,Taylor:2012dr,EJJN2, EJJN1,Esole:2015xfa,Fullwood:SVW,FH2}. 
Crepant resolutions of a Weierstrass model are relative minimal models in the sense of Mori's program \cite{Matsuki}. Different crepant resolutions of the same Weierstrass model are  connected to each other by a sequence of flops.

\subsection{$G$-models and F-theory}
A classical problem in the study of elliptic fibrations is understanding the geometry of the crepant resolutions of Weierstrass models and their flop transitions.
A natural set of singular Weierstrass models to start with are the $G$-models. 
The constructions of $G$-models are deeply connected to the classification of singular fibers of Weierstrass models and provide an interesting scene to explore higher dimensional elliptic fibrations with  a view inspired by their applications to physics. We follow the definitions and notation of Appendix C  of \cite{MMP1}. 
The framework of $G$-models   naturally includes a geometric formulation of basic notions of representation theory such as the theory of root systems and weights of representations. 
 The data characterizing a $G$-model can be understood in the framework of gauge theories, which provides a natural language to talk about the geometry of these elliptic fibrations. 
 F-theory enables a description of gauge theories in string theory and M-theory via geometric engineering based on elliptic fibrations
\cite{Vafa:1996xn,Morrison:1996na,Morrison:1996pp,Bershadsky:1996nh}. 
The data of a gauge theory that can be extracted from an elliptic fibration are its  Lie algebra, its  Lie group, and the set of irreducible representations defining how charged particles transform under the action of the gauge group.

In F-theory, the Lie algebra  is determined by the dual graphs of the fibers over the generic points of the irreducible components of the discriminant locus of the elliptic fibration. 
 The Mordell-Weil group of the elliptic fibration is  conjectured to be isomorphic to the first homotopy group of the gauge group \cite{deBoer:2001wca}. 
 Hence, the Lie group depends on both the singular fibers and the Mordell-Weil group of the elliptic fibration. 
In F-theory, the singular fibers responsible for non-simply laced Lie algebras are not affine Dynkin diagrams, but  twisted affine Dynkin diagrams, as presented in Table \ref{Table:DualGraph} and Figure \ref{figure.F4}, respectively, on pages \pageref{figure.F4} and \pageref{Table:DualGraph}. These twisted affine Dynkin diagrams are the Dynkin duals of the corresponding affine Dynkin diagrams. 
The Dynkin dual of a Dynkin diagram is obtained by inverting all the arrows. In the language of Cartan matrices, two Dynkin diagrams  are dual to each other if their Cartan matrices are transposes of each other.
  The Langlands duality interchanges  B$_n$ and C$_n$,  but preserves all the other simple Lie algebras. 
 For affine Dynkin diagrams, only the ADE series are preserved under the Langlands duality. In particular, the Langlands duals of $\widetilde{\text{B}}_n$,
$\widetilde{\text{C}}_n$, $\widetilde{\text{G}}_2$ and  $\widetilde{\text{F}}_4$ are respectively denoted in the notation of Carter as 
$\widetilde{\text{B}}^t_n$,
$\widetilde{\text{C}}^t_n$, $\widetilde{\text{G}}^t_2$ and  $\widetilde{\text{F}}^t_4$.

\subsection{F$_4$-models: definition and first properties}

One of the major achievements of  F-theory  is the geometric engineering of exceptional Lie groups.  These elliptic fibrations play an essential role in the study of superconformal field theories even in absence of a Lagrangian description. 
 The study of non-simply laced Lie algebras in F-theory started in  May of 1996 during the second string revolution  with  a paper  of  Aspinwall and Gross \cite{Aspinwall:1996nk}, followed shortly afterwards by the classic F-theory paper of Bershadsky, Intriligator, Kachru, Morrison, Sadov, and Vafa  \cite{Bershadsky:1996nh}.  M-theory compactifications giving rise to non-simply laced gauge groups are studied  in \cite{IMS,Diaconescu:1998cn}.

In this paper,  we study the geometry of F$_4$-models, namely,  $G$-models with $G=\text{F}_4$, the exceptional simple Lie group of rank $4$ and dimension $52$. 
  F$_4$ is a  simply connected and non-simply laced  Lie group. 

An F$_4$-model is mathematically constructed as follows.
 Let $B$ be a smooth projective variety of dimension two or higher. Let $S$ be an effective Cartier divisor in $B$ defined as the zero scheme of a section $s$ of a line bundle $\mathscr{S}$.  
 Since $S$ is smooth, the residue field of its generic point is a discrete valuation ring. We denote the valuation with respect to $S$ as  $v_S$. 
An F$_4$-model is defined by the crepant resolution of the following  Weierstrass model 
\begin{equation}\label{Eq:Defining}
 y^2z=x^3+ s^{3+\alpha}a_{4,3+\alpha} x z^2 + s^4 a_{6,4} z^3, \quad \alpha\in\mathbb{Z}_{\geq 0},
\end{equation}
where $v_S(a_{6,4})=0$, and either $v_S(a_{4,3+\alpha})=0$ or $a_{4,3+\alpha}=0$.
We assume that  $a_{6,4}$ is generic; in particular, $a_{6,4}$ is not a perfect square modulo $s$. This ensures that the generic fiber is of type  IV$^{*\text{ns}}$ rather than  IV$^{*\text{s}}$. 

 The defining equation \eqref{Eq:Defining} is of type   (c6) in N\'eron's classification of minimal Weierstrass models, defined  over a perfect residue field of characteristic different from $2$ and $3$ \cite{Neron}. 
This corresponds to Step 8 of Tate's algorithm \cite{Tate}. Hence,   the geometric fiber over the generic point of $S$ is of Kodaira type IV$^*$. As discussed earlier, the generic fiber over $S$  has a dual graph that is  the twisted Dynkin diagram $\widetilde{\text F}^t_4$, which is  the Langlands dual of the affine Dynkin diagram  $\widetilde{\text F}_4$.The structure of the generic fiber of the F$_4$-model is due to the arithmetic restriction that $a_{6,4}$ is not a  perfect square modulo $s$. \\

 The   innocent arithmetic condition which characterizes a F$_4$-model has surprisingly deep topological implications for the elliptic fibration. 
When  $a_{6,4}$  is a  perfect square modulo $s$, the elliptic fibration is called an E$_6$-model; the generic fiber  IV$^{*\text{ns}}$ is replaced by the fiber  IV$^{*\text{s}}$ whose dual graph  is the  affine Dynkin diagram $\widetilde{\text{E}}_6$.  
Both E$_6$-models and F$_4$-models share the same geometric fiber over the generic point of $S$; however, the generic fiber type is different. 
Both are characterized by  Step 8 of Tate's algorithm, but while  the  generic fiber of an E$_6$-model is already made of geometrically irreducible components, the generic fiber of an F$_4$-model requires a quadratic field extension to make all its components  geometrically irreducible.   The Weierstrass model of an E$_6$-model has more complicated singularities, and thus requires a  more involved crepant resolution. 
In  a sense, the F$_4$-model is a more rigid  version of the E$_6$-model as some fibral divisors are glued together by the arithmetic condition on $a_{6,4}$. 
 It follows from the Shioda-Tate-Wazir theorem \cite{Wazir} that the Picard number of an E$_6$-model is bigger than the one for an F$_4$-model. Furthermore, their Poincar\'e-Euler characteristics are also different, as was recently analyzed in \cite{MMP1}. A crepant resolution of the Weierstrass model of an E$_6$-model has fourteen distinct minimal models \cite{Hayashi:2014kca}, while that of an F$_4$-model  has only one \cite[Theorem 1.27]{EJJN1}. An F$_4$-model is a flat elliptic fibration, while this is not  the case for an E$_6$-model of dimension four  or higher, as certain fibers over codimension-three points contain rational surfaces \cite{Kuntzler:2012bu}. 
 F$_4$-models are also studied from different points of view in \cite{Braun:2014oya,Bershadsky:1996nh,Bonora:2010bu,Esole:2012tf}.

The discriminant locus of the elliptic fibration \eqref{Eq:Defining} is 
\begin{align*}
\Delta= s^8 (4  s^{1+3\alpha} a_{4,3+\alpha}^3+27 a_{6,4}^2).
\end{align*} 
The discriminant is composed of two irreducible components not intersecting transversally. The first one is the divisor $S$, and the fiber above its generic point is of type IV$^{*\text{ns}}$. 
 The second component is a singular divisor, and the fiber over its generic point is a nodal curve, i.e. a Kodaira fiber of type I$_1$. The fiber I$_1$ degenerates to a cuspidal curve over $a_{4,3+\alpha}=a_{6,4}=0$. 
 The two components of the discriminant locus collide at $s=a_{6,4}=0$, where we expect the singular fiber IV$^{*\text{ns}}$ to degenerate, further producing  a non-Kodaira fiber.

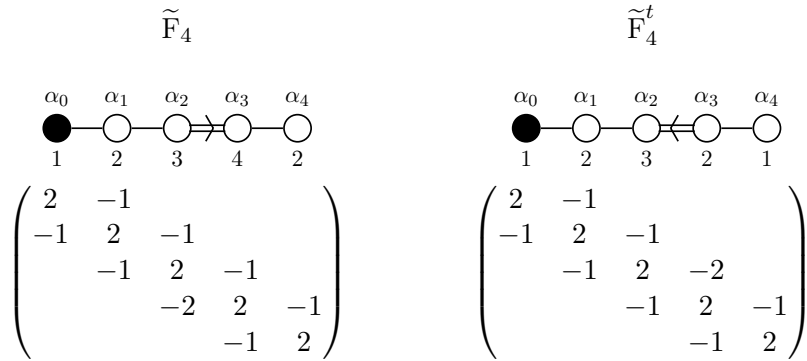
\begin{figure}[htb]
\begin{center}
\begin{tabular}{ccc}
$\widetilde{\text{F}}_4$& \phantom{xxxx}&$\widetilde{\text{F}}^t_4$\\
 \\
\scalebox{.8}{$\begin{array}{c}{
\begin{tikzpicture}
				\node[draw,circle,thick,scale=1.25,fill=black,label=below:{$1$}, label=above:{$\alpha_0$}] (1) at (0,0){};
				\node[draw,circle,thick,scale=1.25,label=below:{$2$},label=above:{$\alpha_1$}] (2) at (1,0){};
				\node[draw,circle,thick,scale=1.25,label=below:{$3$}, label=above:{$\alpha_2$}] (3) at (2,0){};
				\node[draw,circle,thick,scale=1.25,label=below:{$4$},label=above:{$\alpha_3$}] (4) at (3,0){};
				\node[draw,circle,thick,scale=1.25,label=below:{$2$}, label=above:{$\alpha_4$}] (5) at (4,0){};
				\draw[thick] (1) to (2) to (3);
				\draw[thick]  (4) to (5);
				\draw[thick] (2.2,0.05) --++ (.6,0);
				\draw[thick] (2.2,-0.05) --++ (.6,0);
				\draw[thick]
					(2.6,0) --++ (120:.25)
					(2.6,0) --++ (-120:.25);
	\end{tikzpicture}}
	\end{array}
	$}
	
	&&
	\scalebox{.8}{
	$\begin{array}{c}
	{\begin{tikzpicture}
				\node[draw,circle,thick,scale=1.25,fill=black,label=below:{1}, label=above:{$\alpha_0$}] (1) at (0,0){};
				\node[draw,circle,thick,scale=1.25,label=below:{2}, label=above:{$\alpha_1$}] (2) at (1,0){};
				\node[draw,circle,thick,scale=1.25,label=below:{3}, label=above:{$\alpha_2$}] (3) at (2,0){};
				\node[draw,circle,thick,scale=1.25,label=below:{2}, label=above:{$\alpha_3$}] (4) at (3,0){};
				\node[draw,circle,thick,scale=1.25,label=below:{1}, label=above:{$\alpha_4$}] (5) at (4,0){};
				\draw[thick] (1) to (2) to (3);
				\draw[thick]  (4) to (5);
				\draw[thick] (2.2,0.05) --++ (.6,0);
				\draw[thick] (2.2,-0.05) --++ (.6,0);
				\draw[thick]
					(2.4,0) --++ (60:.25)
					(2.4,0) --++ (-60:.25);
			\end{tikzpicture}}\end{array}$}
			\\
		$\begin{pmatrix}
		2 & -1& & & \\
	          -1 & 2&-1 & & \\
                 & -1& 2&-1 & \\
		& & -2&2 &-1 \\
	     & & & -1&2 
		\end{pmatrix}
		$	& &
		$\begin{pmatrix}
		2 & -1& & & \\
	          -1 & 2&-1 & & \\
                 & -1& 2&-2 & \\
		& & -1&2 &-1 \\
	     & & & -1&2 
		\end{pmatrix}
		$
		\\
			\end{tabular}
	\end{center}
	\caption{ {\bf Affine Dynkin diagram $\widetilde{\text{F}}_4$ vs. twisted affine Dynkin diagram $\widetilde{\text{F}}^t_4$}. 
	Their Cartan matrices are transposes of each other. Since the matrices are not symmetric, taking the transpose means inverting the arrow of the Dynkin diagram and changing the multiplicities of the nodes. 
	These matrices have rank four and therefore a kernel of dimension one. The normalization of the zero direction in terms of relatively prime integers gives  the multiplicities of the nodes of the Dynkin diagram.
	 In the notation of Kac,  $\widetilde{\text{F}}^t_4$ is denoted as $\widetilde{\text{E}}_6^{(2)}$ and $\widetilde{\text{F}}_4$ is denoted as
	 $\widetilde{\text{F}}_4${\color{blue},} or sometimes $\widetilde{\text{E}}_6^{(1)}$. 
	 The dual graph that appears in the theory of elliptic fibration is $\widetilde{\text{F}}^t_4$ and never $\widetilde{\text{F}}_4$	. 
	 	}\label{figure.F4}
		\end{figure}

\subsection{Representations associated to an F$_4$-model and flops. }

 To      determine the irreducible representations associated with a given $G$-model, we adopt the approach of Aspinwall and Gross  \cite{Aspinwall:1996nk}, which is deeply rooted in geometry.  The representation associated to a $G$-model  
  is  characterized by  its weights, which are geometrically given by  the intersection numbers, with a negative sign,   of the fibral divisors with  curves appearing over codimension-two loci over which the generic fiber IV$^{*\text{ns}}$ degenerates. 
 This approach,  also used in \cite{IMS,Aspinwall:2000kf, Marsano} and closely related to the approach of  \cite{Morrison:2011mb}, transcends its application to physics and allows an intrinsic determination of a representation for any $G$-model. 
   The representation induced by the weights of vertical curves over codimension-two points is not always  physical as it is possible that no hypermultiplet is charged under that 
  representation. In such a case, the representation is said to be ``frozen" as discussed in \S \ref{sec:frozen}.    

  For F$_4$-models, we find that over $V(s,a_{6,4})$, the fiber IV$^{*\text{ns}}$ degenerates to a non-Kodaira fiber of type $1-2-3-4-2$. This fiber consists of  a chain of rational curves intersecting transversally; each number gives the multiplicity of the corresponding rational curve. We refer to this  non-Kodaira fiber by the symbol  IV$^{*}_{(2)}$.  The last two nodes, of multiplicities $4$ and $2$,  have weights in the  representation $\mathbf{26}$ of F$_4$, namely 
   $\boxed{0\ 1\  -2 \  1 }$ and $\boxed{0\ 0\  1 \  -2 }$ as illustrated in  Figure  \ref{Figure:Degeneration} on page \pageref{Figure:Degeneration}.
  The representation $\mathbf{26}$ is a quasi-minuscule fundamental representation comprised of  two zero weights and 24 non-zero weights that form a unique Weyl orbit. 
  The possibility of flops between different crepant resolutions of the same singular Weierstrass model  can be explained by the relative minimal model program \cite{Matsuki}.

  For each non-simply laced Lie algebra $\mathfrak{g}_0$, there is a specific  quasi-minuscule representation $\mathbf{R}_0$ 
 defined by the branching rule $\mathfrak{g}=\mathfrak{g}_0\oplus\mathbf{R}_0$, where $\mathfrak{g}_0$ is defined by a  folding of $\mathfrak{g}$ of degree $d$.
The degree $d$ is the ratio of the squared lengths  of  long roots and short roots of $\mathbf{R}_0$. In other words, except for G$_2$, for which $d=3$, all other non-simply laced Lie algebras have $d=2$.
 This branching rule is related to the arithmetic degeneration $K^{\text{ns}}\to K^{\text{s}}$, where $K$ is a Kodaira fiber.  
For F$_4$, we have $\text{E}_6=\text{F}_4\oplus \mathbf{26}$, and the representation $\mathbf{R}_0$ coincides with the representation we find by computing weights of the curves over 
$V(s,a_{6,4})$.

  \subsection{Counting hypermultiplets:  Witten's genus formula}\label{Sec:counting}
 M-theory compactified on an elliptically fibered Calabi-Yau threefold  $Y$ gives rise  to  a five-dimensional supergravity theory with eight supercharges   coupled to $h^{2,1}(Y)+1$ neutral hypermultiplets and $h^{1,1}(Y)-1$ vector multiplets \cite{Cadavid:1995bk}. Taking into account the graviphoton, there are a total of $h^{1,1}(Y)$ gauge fields. 
 The  kinetic terms, the Chern-Simons coefficients of the vector multiplets, and the graviphoton are all completely determined by the intersection ring of the Calabi-Yau variety.

   Witten determined the number of states appearing when a curve collapses to a point using a quantization argument \cite{Witten:1996qb}. 
 In particular, he showed that the number of hypermultiplets transforming in the adjoint representation  is  the genus of the curve $S$ over which the gauge group is localized.
  Aspinwall, Katz, and Morrison subsequently applied Witten's quantization argument to  the case of non-simply laced  groups in \cite{Aspinwall:2000kf}. 
 If $S^{\prime}$ is a $d$ cover of $S$ with ramification divisor $R$, then 
   number $n_{{ \mathbf{R}_0}}$ of hypermultiplets transforming in  the representation $\mathbf{R}_0$ is given by \cite{Aspinwall:2000kf}
 \begin{equation}
 n_{{ \mathbf{R}_0}}=g'-g,\nonumber
 \end{equation}
where $g'$ is the genus of $S'$ and $g$ is the genus of $S$. 
   This method is consistent with the six-dimensional anomaly cancellation conditions \cite{GM1}.
Computing $g'-g$ is a classical exercise whose answer is given by the following theorem. 
   \begin{thm}[Riemann-Hurwitz, {see \cite[Chap. IV, Cor. 2.4 and Example 2.5.4.]{Hartshorne}}] \label{prop:genusCover}
Let $f:S'\to S$ be a finite, separable morphism of curves of degree $d$ branched with ramification divisor $R$. If $g'$ is the  genus of $S'$, $g$ is the genus of $S$,  and $R$ is the ramification divisor, then 
$$g'-g=(d-1)( g-1) +\frac{1}{2} \textnormal{deg}\ R.$$
\end{thm}
Physically, this is the number of charged hypermultiplets in the representation $\mathbf{R}_0$, as expected from Witten's quantization argument\footnote{ This is a direct application of 
the Riemann-Hurwitz's theorem (Theorem \ref{prop:genusCover}) in the context of Witten's genus for non-simply laced groups  \cite{Aspinwall:2000kf}. See 
section 3 of \cite{GM2}. 
 }:
\begin{equation}
n_{\mathbf{R}_0}=(d-1)\Big( g-1\Big)+\frac{1}{2} \textnormal{deg}\  R.
\end{equation}
For an F$_4$-model, $\mathbf{R}_0=\mathbf{26}$,  the ramification locus $R$ is  $V(s,a_{6,4})$, and its degree is
\begin{equation}
\textnormal{deg}\ R= 12 (1-g)+2 S^2.\nonumber
\end{equation}
Hence, by Witten's argument, we expect for a generic F$_4$-model the following multiplicities: 
\begin{equation}
n_{\mathbf{52}}=g, \quad 
n_{\mathbf{26}}= 5(1-g) +S^2.\nonumber
\end{equation}

We later derive  in Theorem \ref{Thm:n52andn26} the number of matter representations from a direct comparison of the triple intersection numbers and the one-loop  Intrilligator-Morrison-Seiberg prepotential of a five-dimensional gauge theory \cite{IMS}. 
This matches exactly the number $n_{\bf{26}}$ derived by Witten's genus formula. This provides a  confirmation of the number of charged hypermultiplets from a purely five-dimensional point of view, thereby avoiding a six-dimensional argument  based on cancellations of anomalies \cite{GM1} and the  subtleties of the Kaluza-Klein circle compactification \cite{Bonetti:2013ela,Grimm:2015zea}.

\subsection{Arithmetic versus geometric degenerations}

Given an elliptic fibration $\varphi:Y\to B$, if  $S$ is an irreducible component of the discriminant locus, the generic fiber over $S$ can degenerate further over subvarieties of $S$. 
 We distinguish between two types of degenerations \cite{MMR.I}.  A degeneration is said to be {\em arithmetic} if it modifies the type of the fiber without changing the type of the geometric fiber. 
 A degeneration is said to be {\em geometric} if it modifies the geometric type of the fiber.  

\begin{exmp}[
Arithmetic degeneration] Let $K$ be a Kodaira fiber. Then, $K^{\text{ns}}\to K^{\text{s}}$ or $K^{\text{ns}}\to K^{\text{ss}}$ are arithmetic degenerations.
\end{exmp}

\begin{exmp}\label{exmp:Kd}
[Geometric degeneration] Let $K$ be a Kodaira fiber. Consider the non-split fiber $K^{\text{ns}}$. Denote by $K_{(d)}$ the non-Kodaira fibers defined in the limit where $d$ non-split curves of $K^{\text{ns}}$ coincide. 
Then the fiber $K^{\text{ns}}\to K_{(d)}$ is a geometric degeneration. 
\end{exmp}
 Arithmetic degenerations and geometric degenerations are respectively responsible for non-localized and localized matter in physics.    
 For example, in the case of an F$_4$-model over a base of dimension three or higher, we have the degeneration $\textnormal{IV}^{*\text{ns}}\to \textnormal{IV}^{*\text{s}}$ over the intersection of $S$ with any double cover of $a_{6,4}$ as it is clear from the explicit resolution of singularities. 
 Such a double cover has equation $\zeta^2=a_{6,4}$ where $\zeta$  is a section of $\mathscr{L}^{\otimes 3}\otimes \mathscr{S}^{\otimes 2}$. 
We can have such an arithmetic degeneration over any point of $S$ away from $V(s,a_{6,4})$. 
 
  In the case of an F$_4$-model, we get a fiber of type $1-2-3-4-2$, which is of the type IV$^*_{(2)}$ discussed in Example \ref{exmp:Kd}. 
 The geometry of the fiber shows explicitly that  localized matter fields at $V(s,a_{6,4})$ are in the representation $\mathbf{26}$  and do not come from an enhancement  $F_4\to  E_6$. That is because over the locus 
$V(s,a_{6,4})$, the generic fiber is the  non-Kodaira fiber $1-2-3-4-2$. Such  a fiber  can only be seen as the result of an enhancement of type IV$^{*\text{ns}}$ (F$_4$) to either an incomplete III$^*$ (E$_7$) or an incomplete II$^*$ (E$_8$), depending on the valuation of  the Weierstrass coefficient $a_4$. We have either the enhancement 
 $\text{F}_4\longrightarrow \text{E}_7$ if $v(a_4)=3$, or $\text{F}_4\longrightarrow \text{E}_8$ if $v(a_4)\geq 4$.
\footnote{The fiber $1-2-3-4-2$ also appears in Miranda's models at  the transverse collision II+IV$^*$ where it is presented as a contraction of a fiber of type II$^*$ (E$_8$), see Table 14.1 on page 130 of \cite{Miranda.smooth}. }  
\subsection{Frozen representations }\label{subsec:frozen}
As we have explained before, we identify representations by their weights, and we compute the weights geometrically by the intersection of fibral divisors with vertical curves located over codimension-two points  on  the base. 
It is important to keep in mind that the presence of a given weight is a necessary condition  but  not  a sufficient condition  for the existence of hypermultiplets transforming under  the corresponding representation. 
We always have to keep in mind  that the geometric representations deduced  by the weights of vertical curves over codimension-two points are not necessarily  carried by physical states. 
 A  representation $\mathbf{R}$ deduced geometrically on an elliptic fibration is  said to be {\em frozen} 
when the elliptic fibration has vertical curves (over a codimension-two locus of the base) carrying the weights of the  representation $\mathbf{R}$, but  no hypermultiplet is charged under the representation $\mathbf{R}$.

It is known that the adjoint representation is frozen when the gauge group is on a curve of genus zero \cite{Witten:1996qb}. 
However, it is less appreciated that other representations can also be frozen. A natural candidate is the   representation $\mathbf{R}_0$ discussed in \S\ref{Sec:counting} for the case of a non-simply laced gauge group.
  We will discuss the existence of a frozen representation for  an F$_4$-model in \S\ref{sec:frozen}. 
  In particular, Theorem \ref{Prop:frozen} asserts that the representation $\mathbf{26}$ is frozen if and only if the curve $S$ has genus zero and self-intersection $-5$.

It is a folklore theorem of D-brane model building that the number of representations appearing at the transverse collision of two branes is the number of collision points. 
In other words, one would expect one hypermultiplet for each intersection points\footnote{ There are some subtleties: when over  a collision point the same weight is induced by $n$ distinct  vertical curves, we can have up $n$ hypermultiplets localized at that point. See for example \cite{GM1,Katz:1996xe, Arras:2016evy}.}. While this is true for localized matter fields, it is not usually true when 
the same representation  appears both as localized
 and  non-localized. 
As a rule of thumb, when  localized and non-localized matter fields transforming in  the same representation coexist, the number of representations is given by assuming that all the matter is non-localized \cite{Morrison:2012np}. 

The notion of frozen representation  discussed in this paper should not be confused with the frozen singularities of ref. \cite{Tachikawa:2015wka,Atiyah:2001qf,deBoer:2001wca}.

  \subsection{Summary of results}

The purpose of  this paper is to study the geometry of F$_4$-models.  We define an F$_4$-model as an elliptic fibration over a smooth variety of dimension two or higher such that the singular fiber over the  generic point of a chosen  irreducible Cartier divisor $S$ is of type  IV$^{*\text{ns}}$, and the fibers are irreducible (smooth elliptic curves, type II or type I$_1$) away from $S$. 
 Such an F$_4$-model is realized by a crepant resolution of  a Weierstrass model, whose coefficients have valuations with respect to $S$, matching the generic case of Step 8 of Tate's algorithm. 
 An  F$_4$-model can  is always birational  to  a singular short Weierstrass equation  
  whose singularities are due to the valuations $v_S(c_4)\geq 3$ and $v_S(c_6)=4$ of its coefficients. 
Such a singular Weierstrass equation can be  traced back to N\'eron's seminal paper \cite{Neron}  where it corresponds to type c6.
The road map to the rest of the  paper is the following. 

 In section \ref{Sec:Coventions}, we summarize our most common conventions and basic definitions.  We also discuss Step 8 of Tate's algorithm, which characterizes the Weierstrass model of F$_4$-models. 

   In section \ref{Sec:Resolution} (see Theorem \ref{Thm:blowups} on page \pageref{Thm:blowups}), we  present a crepant resolution of the singular Weierstrass model, giving a flat fibration. 
The resolution is given by a sequence of four blowups with centers that are regular monomial ideals. 

 In section  \ref{Sec:Fiber}, we analyze in details the degeneration of the singular fiber and determine the geometry of the fibral divisors (see Theorem \ref{Thm:fibralGeom} on  page \pageref{Thm:fibralGeom} and  Figure \ref{Fig:FibralDiv} on page  \pageref{Fig:FibralDiv}). 
The generic fiber over $S$ degenerates along $V(a_{6,4})\cap S$ to produce a non-Kodaira fiber of type $1-2-3-4-2$.
 This  non-Kodaira fiber appears as  an incomplete Kodaira fiber of type III$^*$  or II$^*$ resulting from the (non-transverse) collision of the divisor $S$ with the remaining factor of the discriminant locus.  
Such a collision is not of Miranda-type since it involves two
 fibers of different $j$-invariants\cite{Miranda.smooth}. 

We show that the fibral divisors $D_3$ and $D_4$ corresponding to the root $\alpha_3$ and $\alpha_4$ of the F$_4$ Dynkin diagram are {\bf not} $\mathbb{P}^1$-bundles over the divisor $S$, but rather double covers of $\mathbb{P}^1$-bundles over $S$, with ramification locus  $V(a_{6,4})$.  
The geometry of these fibral divisors is illustrated in Figure \ref{Fig:FibralDiv}.
 The difference is important since it affects  the computation of triple intersection numbers and the degeneration of the fibers in codimension two, which is responsible for the appearance of 
 weights of the  representation $\mathbf{26}$. 
  We use the Stein factorization to have more control on the geometry of $D_3$ and $D_4$. 
 Consider the morphism $f: D_3\to S$.  The  geometric generic fiber is not connected and consists of two rational curves. 
Since the morphism is proper, we consider its  Stein factorization  $D_3\overset{f'}{\longrightarrow }S'\overset{\pi}{\longrightarrow } S$. 
By definition, the morphism $\pi: S'\to S$ is a finite map of degree two; each geometric point of the fiber represents
a connected component of the fiber of $D_3\to S$.  
The morphism $f': D_3\to S'$ has connected fibers that are all smooth rational curves. Hence, $f': D_3\to S'$ gives  $D_3$ the structure of a $\mathbb{P}^1$-bundle over $S'$ rather than over $S$.

We  determine in  \S \ref{Sec:Rep} the geometric weights that identify the representation naturally associated to the degeneration of the generic fiber over the codimension-two loci.  
The last two nodes of the fiber $1-2-3-4-2$ are responsible for generating the representation $\mathbf{26}$.

 In section \ref{Sec:Topology}, we compute the following topological invariants: the Euler characteristic of the elliptic fibration, the Hodge numbers in the Calabi-Yau threefold case,  the double and triple intersection numbers of the fibral divisors
 (Theorem \ref{Thm:degeneration} on page \pageref{Thm:degeneration}), and the linear form induced in the Chow ring by the second Chern class in the case of a Calabi threefold. 

 In section  \ref{Sec:Mtheory}, we leverage our understanding of the geometry to make a few statements on the physics of F$_4$ gauge theories in different dimensions. 
We specialize to the case of a Calabi-Yau threefold and consider an M-theory compactified on an F$_4$-model that is also a Calabi-Yau threefold. 
We compute the number of  hypermultiplets in the adjoint and fundamental representations using the triple intersection numbers.  We then match them to the coefficients of the 
five-dimensional cubic prepotential computed at the one-loop level in \cite{IMS}. 
 We check that the resulting spectrum is consistent with an anomaly-free parent six-dimensional gauge theory. 
Finally,  in \S\ref{sec:frozen} we discuss in detail the existence of frozen representations for an F$_4$-model.

\section{Basic conventions and  definitions}\label{Sec:Coventions}

We work over the complex numbers and assume that $B$ is a nonsingular variety, $\mathscr{L}$ is a line bundle over $B$, and $S=V(s)$ is a smooth irreducible subvariety of $B$ given by the zero scheme of a section $s$ of a line bundle $\mathscr{S}$. 
We use the conventions of Carter  and denote an affine Dynkin diagram by $\tilde{\mathfrak{g}}$, where $\mathfrak{g}$ is the Dynkin diagram of a simple Lie algebra \cite{Carter}.
We write $\tilde{\mathfrak{g}}^t$ for the twisted Dynkin diagram whose Cartan matrix is the transpose of the Cartan matrix of $\tilde{\mathfrak{g}}$. This notation is only relevant when $\mathfrak{g}$  is not simply laced, that is, for $\mathfrak{g}=$ G$_2$, F$_4$, B$_{3+k}$, or C$_{2+k}$. 
Given a vector bundle $\mathscr{V}$, we denote by $\mathbb{P}[\mathscr{V}]$ the projective bundle of lines of $\mathscr{V}$. In intersection theory, we follow the conventions of Fulton \cite{Fulton.Intersection}. 
We denote the geometric fibers of an elliptic surface by Kodaira symbols. To denote a generic fiber, we decorate the Kodaira fiber by an index ``ns'' , ``ss'' , or ``s'' that characterizes the degree of the field extension necessary to move from the generic fiber to the geometric generic fiber.

\subsection{Geometric weights}

Let $\varphi: Y\to B$ be a smooth  flat elliptic fibration whose discriminant has a unique component $S$ over which the generic fiber is reducible with dual graph the affine Dynkin diagram $\tilde{\mathfrak{g}}^t$. 
We denote the   irreducible components of the generic fiber over $S$ as  $C_a$. 
  If $\mathfrak{g}$ is not simply laced, the curves $C_a$ are not all geometrically irreducible. 
 Let   $D_a$ be the fibral divisors over $S$. By definition, $\varphi^* (S)=\sum_a m_a D_a$. 
   The curve $C_a$ can also be thought of as the generic fiber of $D_a$ over $S$.
Let $C$ be a vertical curve of the elliptic fibration. 

We define the weight of a vertical $C$ with respect to a fibral divisor $D_a$ as the intersection number 
$$
\varpi_a(C):=- \int_Y D_a \cdot C.
$$
 Using intersection of curves with fibral divisors to determine a representation from an elliptic fibration is a particularly robust algorithm  since the intersections of divisors and curves are well-defined even in the presence of singularities \cite{Fulton.Intersection}. 
We can ignore the intersection number of the divisor touching the section of the elliptic fibration as it is fixed in terms of the others, thus allowing us to write 
\begin{align*}
\varpi(C)=\boxed{\varpi_1(C)\  , \    \varpi_2(C)\  , \  \cdots \   ,\     \varpi_n(C)}.
\end{align*}
We interpret $\varpi(C)$ as the weight of the vertical curve $C$ in the basis of fundamental weights. This interpretation implies that the fibral divisors play the role of   co-roots of $\mathfrak{g}$, while vertical curves are identified with elements of the weight lattice 
of $\mathfrak{g}$.

The notion of {\em a saturated set of weights} is introduced in  Bourbaki (Groups and Lie Algebras,  Chap.VIII.\S 7. Sect. 2.) and provides 
the algorithm to determine a representation from a subset of its weights. See \cite{MMR.I} for more details.

  It follows from the general theory of elliptic fibration that the intersection of the generic fibers with the fibral divisors gives the invariant form of the affine Lie algebra $\tilde{\mathfrak{g}}^t$, where $\mathfrak{g}$ is the Lie algebra of $G$. 
The matrix  $\varpi_a(C_b)$ is the invariant form of the Lie algebra $\tilde{\mathfrak{g}}$ in the normalization  where short roots have diagonal entries $2$.

\subsection{Step 8 of Tate's algorithm}

 We follow the notation of Fulton \cite{Fulton.Intersection}. The terminology is borrowed from \cite{MMP1}.  
 Let $Y_0\longrightarrow B$ be a  Weierstrass model over a smooth base $B$, in which we choose a smooth Cartier divisor $S\subset B$. 
The local ring  $\mathscr{O}_{B,\eta}$ in $B$ of the generic point $\eta$ of $S$ is a discrete valuation ring with valuation $v_S$ given by the multiplicity along $S$. 
Using Tate's algorithm, the valuation of the coefficients of the Weierstrass model with respect to $v_S$ determines the type of the singular fiber over the generic point of $S$. 
Kodaira fibers refer to the type of the geometric fiber over the generic point of irreducible components of the discriminant locus of the Weierstrass model.

An  F$_4$-model describes  the generic case of  Step 8 of  Tate's algorithm, which characterizes the  Kodaira fiber of type IV$^{*\text{ns}}$ in  F-theory notation or IV$^*_{2}$ in the notation of Liu.
By definition, Kodaira fibers classify geometric fibers over the generic point of a component of the discriminant locus of an elliptic fibration. 
When the elliptic fibration is given by a Weierstrass model, this can be expressed in the language of a discrete valuation ring. 
Let $S$ be the relevant component. We assume that $S$ is smooth with generic point $\eta$. 
The local ring  at $\eta$ defines a discrete valuation ring with valuation $v$ that is essentially the multiplicity along $S$. 
We then have the following characterization:
\begin{thm}[Tate's algorithm \cite{Tate}, Step 8]
If  $v(a_1)\geq 1$, $v(a_2)\geq 2$, $v(a_3)\geq 2$, $v(a_4)\geq 3$, $v(a_6)\geq 4$ and the quadric polynomial  
$
Q(T)=T^2 + a_{3,2}  T-a_{6,4}
$
has two distinct solutions, then the  geometric special fiber is of Kodaira type IV$^*$. If the roots of $Q(T)$ are rational in the residue field, the generic fiber is of type  IV$^{*\text{s}}$,  otherwise (if the solutions are not rational in the residue field) the generic fiber is of  type  IV$^{*\text{ns}}$. 
\end{thm}
\begin{rem}
The discriminant of the quadric $Q(T)=T^2 + a_{3,2}  T-a_{6,4}$, is exactly $b_{6,4}$. It follows that the fiber is of type IV$^*$ if and only if $v(b_{6,4})=0$. Moreover, the fiber is of either type   IV$^{*\text{s}}$ or  IV$^{*\text{ns}}$, depending respectively on whether or not $b_{6,4}$ is a perfect square.
\end{rem}
In view of the multiplicities, we can safely complete the square in $y$ and the cube in $x$ and write the Tate equation of a  IV$^{*\text{ns}}$ model as 
$$
y^2 z=x^3 + a_{4,3+\alpha } s^{3+\alpha}  x z^2 + s^4 a_{6,4} z^3, \quad \alpha \in \mathbb{Z}_{\geq 0},
$$
where $a_{6,4}$ is  {\bf not} a perfect square modulo $s$. 

The simplest way to identify a fiber of type IV$^{*\text{ns}}$ is to use the short Weierstrass equation since it does not require performing any translation. 
\begin{thm}[Tate's algorithm\cite{Tate}, Step 8]
$$
\begin{cases}
v(c_4)\geq 3, \quad v(c_6)=4\\
c_6 \quad \text{not a square modulo $s$}
\end{cases}
\iff \text{IV}^{*\text{ns}}.
$$
\end{thm}
The condition on $c_4$ and $c_6$ can be traced back to N\'eron and forces the discriminant to have valuation $8$. N\'eron also points out that a fiber of type  IV$^*$ is uniquely identified by the valuation of its 
 $j$-invariant and its discriminant locus:
\begin{thm}[N\'eron \cite{Neron}]
$v(j)>0 \quad \text{and}\quad v(\Delta)=8\iff \text{IV}^*$. 
\end{thm}
This implies in particular that a fiber of type IV$^*$ has a vanishing $j$-invariant.

\section{Crepant resolution}\label{Sec:Resolution}

Let  $X_0=\mathbb{P}[\mathscr{O}_B\oplus \mathscr{L}^{\otimes 2}\oplus\mathscr{L}^{\otimes 3}]$ be the projective bundle in which the singular Weierstrass model 
 is defined as  a hypersurface.  The tautological line bundle of the projective bundle $X_0$ is denoted $\mathscr{O}(-1)$ and its dual $\mathscr{O}(1)$ has first Chern class   $H=c_1 \big( \mathscr{O}(1) \big)$. 
 
 Let $X$ be a nonsingular variety. 
 Let $Z\subset X$ be a complete intersection defined by the transverse intersection of $r$ hypersurfaces $Z_i=V(g_i)$, where $g_i$ is a section of the line bundle $\mathscr{I}_i$ and $(g_1, \cdots, g_r)$ is a regular sequence. 
 We denote the blowup of a nonsingular variety $X$ along the complete intersection $Z$ by 
 $$\begin{tikzpicture}
	\node(X0) at (0,-.3){$X$};
	\node(X1) at (3,-.3){$\widetilde{X}.$};
	\draw[big arrow] (X1) -- node[above,midway]{$(g_1,\cdots ,g_{r}|e_1)$} (X0);	
	\end{tikzpicture}
	$$
The blowup of $X$ with center $Z$ is the morphism	$f:\widetilde{X}=\mathbf{Proj}_X(\oplus_d {\mathscr{I}}^d)\to X$. The exceptional divisor of $f$ is the pre-image of the center $Z$, that is, $\widetilde{Z}=\mathbf{Proj}_X(\oplus_d {\mathscr{I}}^{d}/{\mathscr{I}}^{d+1})$. The exceptional divisor is $f$-relatively  ample. 
	If  $Z$ is a complete intersection, then $ {\mathscr{I}}/{\mathscr{I}}^2$ is locally free. Hence,  $\mathbf{Sym}^d ( {\mathscr{I}}/{\mathscr{I}^2})={\mathscr{I}}^{d}/{\mathscr{I}}^{d+1} $ and 
	$\widetilde{Z}=\mathbb{P}_X(\mathscr{I}/\mathscr{I}^{2})$. The normal sheaf $N_{\widetilde{Z}|\widetilde{X}}$ is $\mathscr{O}_{\tilde{Z}}(-1)$,   $E_1=c_1 (\mathscr{O}_{\widetilde{Z}}(1))$  is the first Chern class of the exceptional divisor $\widetilde{Z}=V(e_1)$, and $[\widetilde{Z}]=E_1\cap [\widetilde{X}]$.

We abuse notation and use the same symbols for $x$, $y$, $s$, $e_i$ and their successive proper transforms. We also do not write the obvious pullbacks.

\begin{lem}
Let $Z\subset X$ be a  smooth complete intersection of $n+1$ hypersurfaces meeting transversally. Let $Y$ be a hypersurface in $X$ singular along $Z$. If $Y$ has multiplicity $n$ along $Z$, then the blowup of $X$ along $Z$ restricts to a crepant morphism $\overline{Y}\longrightarrow Y$ for the proper transform of $Y$. 
\end{lem}
\begin{proof}
Let $\widetilde{X}=Bl_Z X$  be the blowup of $X$ along $Z$ and $E$ be the class of the exceptional divisor. Then  $c_1(TX)=f^* c_1(TX)-n E$. Since $Y$ has a multiplicity $n$ along $Z$, we have  $f^* Y=\overline{Y}+n E$, where $\overline{Y}$ is the proper transform of $Y$. 
By adjunction, $c_1(\overline{Y})= c_1(\widetilde{X})-\overline{Y}=f^* c_1(X)-f^* Y=f^* c_1(Y)$.    
\end{proof}

\begin{thm}\label{Thm:blowups}
Consider the  following  Weierstrass equation where $S=V(s)$ is a Cartier divisor of the base $B$ :
\begin{align}\nonumber
\mathscr{E}_0:\quad y^2z  =x^3 +s^{3+\alpha} fxz^2 + s^4 gz^3, \quad \alpha\in\mathbb{Z}_{\geq 0}, 
\end{align}
where $f$, $g$, and $s$ are respectively assumed to be generic sections of  $\mathscr{L}^{\otimes 4}\otimes \mathscr{S}^{-\otimes (3+\alpha)}$, 
$\mathscr{L}^{\otimes 6}\otimes \mathscr{S}^{-\otimes 4}$, and $\mathscr{S}$. 
Let  $X_0=\mathbb{P}_B[\mathscr{O}_B\oplus \mathscr{L}^{\otimes 2}\oplus\mathscr{L}^{\otimes 3}]$ be the ambient space in which $\mathscr{E}_0$ is defined. 
The following sequence of blowups provides a crepant resolution of the singular Weierstrass model $\mathscr{E}_0$: 
\begin{equation}\label{Eq:blowups}
 \begin{tikzpicture}
	\node(X0) at (0,0){$X_0$};
	\node(X1) at (2.5,0){$X_1$};
	\node(X2) at (5,0){$X_2$};
	\node(X3) at (8,0){$X_3$};
	\node(X4) at (11,0){$X_4.$};
	\draw[big arrow] (X1) -- node[above,midway]{$(x,y,s|e_1)$} (X0);	
	\draw[big arrow] (X2) -- node[above,midway]{$(y,e_1|e_2)$} (X1);
	\draw[big arrow] (X3) -- node[above,midway]{$(x,e_2|e_3)$} (X2);		
	\draw[big arrow] (X4) -- node[above,midway]{$(e_2,e_3|e_4)$} (X3);
	\end{tikzpicture}
	\end{equation}
We describe this sequence of blowups starting with the projective bundle $X_0$, which serves as the ambient space of the Weierstrass equation. 
The first blowup  $X_1\longrightarrow X_0$ is centered at the regular monomial ideal $(x,y,s)$, where $s$ is a section of $\mathscr{S}=\mathscr{O}_B(S)$. 
The exceptional divisor $E_1$ of the first blowup is a $\mathbb{P}^2$ bundle. 
 The second blowup  $X_2\longrightarrow X_1$, parametrized by $[x:s]$,  is centered along the fiber of $E_1$ defined by the proper transform of $V(y)$ and  its exceptional divisor is $E_2$.
The third blowup $X_3\longrightarrow X_2$ is centered in $E_2$ along the fiber over $V(x)$ and has exceptional divisor $E_3$.
The last blowup   $X_4\longrightarrow X_3$ is centered in $E_3$ along the fiber given by $V(e_2)$ and has exceptional divisor $E_4$.
 
\end{thm}
\begin{proof}
We recall that blowup up of a divisor is an isomorphism away from the singular locus. The Weierstrass model has a singular scheme supported on the ideal $(x,y,s)$. 
\begin{enumerate}
\item \textit{First blowup.} Since the generic point of this ideal is a double point singularity of the Weierstrass model and the ideal has length $3$, blowing up $(x,y,s)$ is a crepant morphism. 
\item \textit{Second blowup.} We are in $X_1$ and the singular locus is supported on $(y, x,e_1)$. At this point, we could choose to blowup again $(x,y,e_1)$ since it is a locus of double points and the ideal has length $3$. 
However, we could also blowup $(y,e_1)$, which is a  non-Cartier Weil divisor. This is clearly crepant since $(y,e_1)$ has length $2$ and multiplicity one. Blowing up this divisor is not an isomorphism since it contains $(x,y,e_1)$, the support of the singular locus.
\item \textit{Third blowup.} We blowup the ideal  $(x,e_2)$, which corresponds to a   non-Cartier Weil divisor  of multiplicity one. 
\item \textit{Fourth blowup.} We finally blowup $(e_2,e_3)$, which is also a non-Cartier Weil divisor of multiplicity one. This is crepant because the ideal has length $2$ and the defining equation has multiplicity one along $(e_2,e_3)$.
\end{enumerate}

 After the fourth blowup, we check using the Jacobian criterion that there are no singularities left. 
We can also simplify  computations by noticing that the defining equation is a double cover and therefore, the singularities should be on the branch locus.  Moreover, certain variables cannot vanish at the same time due to the centers of the blowups. 
In particular, each of $(x,y,s)$, $(y,e_1)$, $(s,e_3)$, $(s,e_4)$, $(x,e_2)$, $(x,e_4)$, and $(e_2, e_3)$ corresponds to the empty set in $X_4$. 
\end{proof}
The divisor classes of the different variables in $X_i$ are given in the following table:\\

\noindent \scalebox{.95}{
\begin{tabular}{|c||   l  |  l  |  l  |  l  |  l |  l |   l  |  l  |}
\hline
& $x$ & $y$ &  $z$ &  $s$ & $e_1$ & $e_2$& $e_3$ &$e_4$\\
 \hline 
 $X_0$& $2L+H$ &  $3L+H$ & $H$ & $S$ & - & -& -& - \\
 \hline
   $X_1$& $2L+H-E_1$ &  $3L+H-E_1$ & $H$ & $S-E_1$ & $E_1$ &- &- &-\\
 \hline 
  $X_2$& $2L+H-E_1$ &  $3L+H-E_1-E_2$ & $H$ & $S-E_1$ & $E_1-E_2$ & $E_2$ &- &-\\
 \hline 
  $X_3$& $2L+H-E_1-E_3$ &  $3L+H-E_1-E_2$ & $H$ & $S-E_1$ & $E_1-E_2$ & $E_2-E_3$ & $E_3$ &-\\
 \hline 
$ X_4$& $2L+H-E_1-E_3$ &  $3L+H-E_1-E_2$ & $H$ & $S-E_1$ & $E_1-E_2$ & $E_2-E_3-E_4$ & $E_3-E_4$ & $E_4$ \\
 \hline 
 \end{tabular}}
 \newline
 
 The proper transform of the Weierstrass model is a smooth elliptic fibration $\varphi:Y\longrightarrow B$
	\begin{align}\nonumber
Y:		e_2 y^2 z&=e_1( e_3^2 e_4 x^3+ e_1 e_2 e_3 (e_1 e_2 e_3 e_4^2)^{\alpha} s^{3+\alpha} a_{4,3+\alpha}xz^2+ e_1 e_2 s^4a_{6,4}z^3).
	\end{align}
The successive relative ``projective coordinates'' for the fibers of  $X_i$ over $X_{i-1}$ are $(i=1,2,3,4)$
	\begin{align}\nonumber
 [e_1 e_2 e_3^2 e_4^3 x : e_1 e_2^2 e_3^2 e_4^4 y : z], ~[ e_3 e_4 x : e_2 e_3 e_4^2 y : s],                        [y:e_1],~[ x:e_2 e_4], ~[e_2 :e_3].
	\end{align}
These projective coordinates are not independent of each other, as we have a tower of  projective bundles defined over subvarieties of projective bundles. 
The interdependence between the different projective bundles are captured  by the following scalings:
	$$
	\begin{aligned}
		X_0/B &\quad [e_1 e_2 e_3^2 e_4^3 x : e_1 e_2^2 e_3^2 e_4^4 y :  z] \\
		X_1/ X_0 &\quad  [\ell_1 ( e_3 e_4 x) : \ell_1 (e_3 e_4^2 e_2 y) : \ell_1 s] \\
		X_2/ X_1 &\quad [\ell_1 \ell_2 y:\ell_1^{-1}\ell_2 e_1]\\
		X_3/X_2 &\quad [\ell_1\ell_3 x:\ell_2^{-1}\ell_3 (e_2 e_4)]\\
		X_4/X_3 &\quad [\ell_4 \ell_3 \ell_2^{-1} e_2 :\ell_4 \ell_3^{-1} e_3]
	\end{aligned}
	$$
	where $\ell_1$, $\ell_2$, $\ell_3$, and $\ell_4$ are used to denote the scalings of each blowup.

\section{Fiber structure}\label{Sec:Fiber}
 In this section, we explore the geometry of the crepant resolution  $Y\to \mathscr{E}_0$ obtained in the previous section.  
Composing with the projection of $\mathscr{E}_0$ to the base $B$, we have a surjective morphism $\varphi:Y\to B$, which is an elliptic fibration over $B$. 
We denote by $\eta$ a generic point of $S$. 
We study in details the generic fiber $Y_\eta$ of the elliptic fibration and its specialization. 
Its dual graph is  the twisted Dynkin diagram $\widetilde{\text{F}}_4^t$, namely, the dual of the affine Dynkin diagram $\tilde{\text F}_4$.
We call $C_a$ the irreducible components of the generic fiber, and $D_a$ the irreducible fibral divisors.
 We can think of $C_a$ as the generic fiber of $D_a$ over $S$. Given a section $u$ of a line bundle, we denote by $V(u)$ the vanishing scheme of $u$. As a set of point, $V(u)$ is defined by the equation  $u=0$.  
If $\cal{I}$ is an ideal sheaf, we also denote by $V(\cal{I})$ its zero scheme.

\subsection{Structure of the generic fiber}
After the blowup, the generic fiber over $S$ is composed of five curves since the total transform of $s$ is $se_1 e_2 e_3 e_4^2$. 
 The irreducible components of the generic curve $Y_\eta$ are the following five curves:  
\begin{align}
C_0 :& \quad s= e_2 y^2-e_1e_3^2 e_4 x^3=0\\
C_1 :&\quad e_1=e_2=0\\
C_2 :& \quad e_2=e_4=0 \\
C_3 :&\quad e_4=  y^2-e_1^2  s^4a_{6,4}z^2=0\\
C_4 :& \quad e_3=  y^2-e_1^2 s^4a_{6,4}z^2=0
\end{align}
Their respective multiplicities are $1$, $2$, $3$, $2$, and $1$.

The curve $C_a$ is the generic fiber of the fibral divisor $D_a$ ($a=0,1,2,3,4$). 
The fibral divisors can also be defined as the irreducible components of $\varphi^* S$:
\begin{equation}
\varphi^* S= D_0 + 2 D_1+ 3 D_2 + 2 D_3 + D_4.
\end{equation}

Furthermore, we have the following relations:
	\begin{align}
		V(s)=D_0, \quad V(e_1) = D_1,\quad  V(e_2) = D_1+D_2,\quad V(e_4)= D_2+D_3,\quad  V(e_3) =D_4.
	\end{align}
Denoting by $E_i$ the exceptional divisor of the $i$th blowup and by $S$ the class of $S$,  we identify the classes  of the  five fibral divisors to be
	\begin{align}\label{eqn:divisorClass}
	\begin{aligned}
		&D_0=S-E_1,\quad D_{1} = E_1-E_2,\quad D_{2} =2 E_2 - E_1-E_3-E_4,\\
		  &D_3 = 2E_4-2E_2+E_1+E_3  ,\quad D_4 = E_3-E_4.
		  \end{aligned}
	\end{align}

The curve $C_0$ is the normalization of a cuspidal curve. 
The curves $C_1$ and $C_2$ are smooth rational curves. 
The curves $C_3$ and $C_4$ are not geometrically irreducible. 
After a field extension that includes the square root of $a_{6,4}$, they split into two smooth rational curves. 
Hence, $D_0$, $D_1$, and $D_2$ are $\mathbb{P}^1$-bundles while 
$D_3$ and $D_4$ are double coverings of $\mathbb{P}^1$-bundles. Geometrically, when $D_3$ and $D_4$  are seen as  families of curves over $S$, $D_3$ and $D_4$ are  families of pair of lines
.

In the next subsection, we determine what these $\mathbb{P}^1$-bundles are  up to an isomorphism.

	\subsection{Fibral divisors}

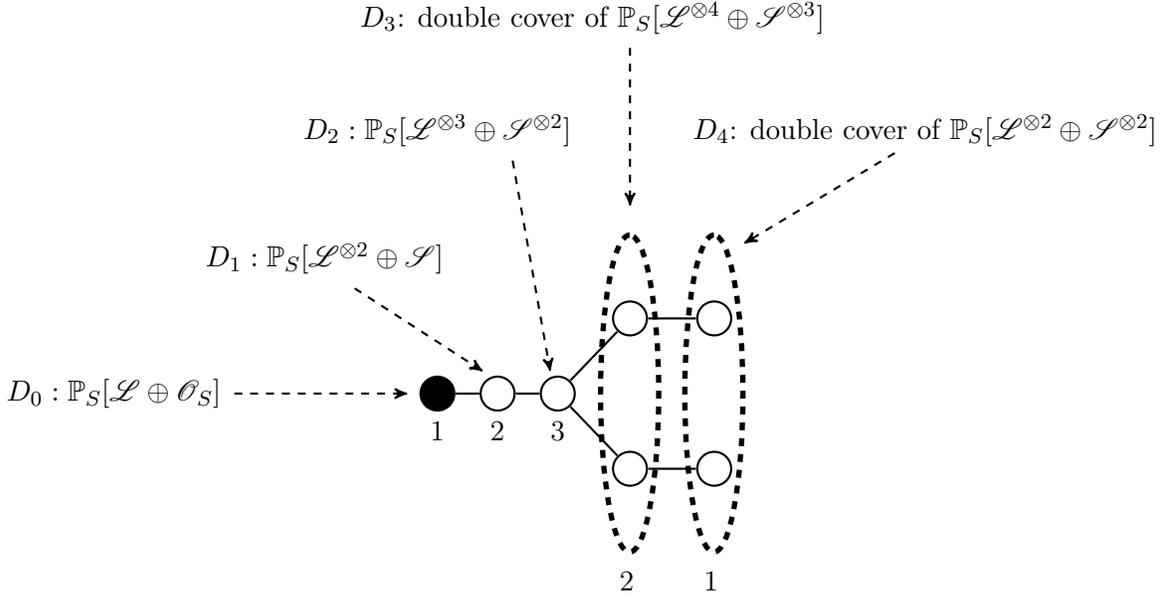
\begin{figure}[htb]
\begin{center}
\begin{tikzpicture}
				\node[draw,circle,thick,scale=1.25,fill=black,label=below:{1}] (0) at (0,0){};
				\node[draw,circle,thick,scale=1.25,label=below:{2}] (1) at (.8,0){};
				\node[draw,circle,thick,scale=1.25,label=below:{3}] (2) at (.8*2,0){};
				\node[draw,circle,thick,scale=1.25] (3) at (.8*3.2,-1){};
				\node[draw,circle,thick,scale=1.25] (4) at (.8*4.6,-1){};
								\node[draw,circle,thick,scale=1.25] (5) at (.8*3.2,1){};
																\node[draw,circle,thick,scale=1.25] (6) at (.8*4.6,1){};
			        \node at (.8*3.15,0)[draw,dashed, line width=2pt, ellipse, minimum width=120pt, minimum height=22pt,rotate=90,yshift=-1pt]{};
				\node at (.8*4.55,0)[draw,dashed, line width=2pt, ellipse, minimum width=120pt, minimum height=22pt,rotate=90,yshift=-1pt]{};
									\node   at (.8*3.15,-2.5) { 2};\node   at (.8*4.55,-2.5) { 1};

				\draw[thick] (0)--(1)--(2)--(3)--(4);
				\draw[thick] (2)--(5)--(6);
								\node (D0) at (-4.3,0) {$D_0:\mathbb{P}_S[\mathscr{L}\oplus\mathscr{O}_S]$};
				\node (D1) at (-1.5,1.8) {$D_1:\mathbb{P}_S[\mathscr{L}^{\otimes 2}\oplus\mathscr{S}]$};
					\node (D2) at (0,3.5) {$D_2:\mathbb{P}_S[\mathscr{L}^{\otimes 3}\oplus\mathscr{S}^{\otimes 2}]$};
						\node (D3) at ($(3)+(-.5,6)$) {$D_3$: double cover of $ \mathbb{P}_S[\mathscr{L}^{\otimes 4}\oplus\mathscr{S}^{\otimes 3}]$};
							\node (D4) at ($(4)+(2.8,4.5)$){$D_4$: double cover of $\mathbb{P}_S[\mathscr{L}^{\otimes 2}\oplus\mathscr{S}^{\otimes 2}]$};
							\draw[->,>=stealth',thick=1mm,dashed]  ($(D0)+(1.6,0)$)--($(0)-(0.4,0)$);
													\draw[->,>=stealth',thick=1mm,dashed]  ($(D1)+(0.4,-.4)$)--($(1)+(-.2,.3)$);
													\draw[->,>=stealth',thick=1mm,dashed]  ($(D2)+(1,-.4)$)--($(2)+(-.1,.3)$);
				\draw[->,>=stealth',thick=1mm,dashed]  ($(D3)+(.5,-.4)$)--($(3)+(0,3.5)$);
				\draw[->,>=stealth',thick=1mm,dashed]  ($(D4)+(-.4,-.3)$)--($(4)+(.4,3)$);
								
	\end{tikzpicture}
\end{center}
		\caption{  
		Fibral divisors of an F$_4$-model as schemes over $S$. See Theorem \ref{Thm:fibralGeom}.
		The fibral divisors $D_0$, $D_1$, and $D_2$ are $\mathbb{P}^1$-bundles over $S$; each of $D_3$ and $D_4$ is a double cover of $S$ branched at $V(s,a_{4,6})$. 
		The generic fiber of $f_i:D_i\to S$  (with $i=3$ or $4$) is not connected and consists of two non-intersecting rational curves. 
		The Stein factorization gives a morphism $f'_i:D_i\to S'$ with connected fibers and a finite morphism  $\pi: S'\to S$ that is a double cover branched at $V(s,a_{4,6})$. 
		The morphism $f'_3:D_3\to S'$ is the $\mathbb{P}^1$-bundle $\mathbb{P}^1_{S'}[\pi^*(\mathscr{L}^{\otimes 2}\oplus\mathscr{S}^{\otimes 3})]\to S'$; the morphism 
		$f'_4:D_4\to S'$ is the $\mathbb{P}^1$-bundle $\mathbb{P}^1_{S'}[\pi^*(\mathscr{L}^{\otimes 2}\oplus\mathscr{S}^{\otimes 2})]\to S'.$
\label{Fig:FibralDiv}				}
		\end{figure}

	In this section, we study the geometry of the fibral divisors. We recall that for a  $\mathbb{P}^1$-bundle, all fibers are  smooth projective curves with no multiplicities. A conic bundle has a discriminant locus, over which the fiber is reducible  when it is composed of two rational curves meeting transversally or is a double line.

In the case of an F$_4$-model, the fibral divisors $D_0$, $D_1$, and $D_2$ are $\mathbb{P}^1$-bundles while $D_3$ and $D_4$ are double covers of $\mathbb{P}^1$-bundles. 
The generic fiber of $D_3$ and $D_4$ is geometrically composed of two non-intersecting rational curves. 
 $D_3$ and $D_4$ have $V(a_{6,4})$ as a  discriminant locus. 
Over the discriminant locus of these double covers, the fiber is composed of a double rational curve.

 We can also simply describe $D_3$ and $D_4$ as flat double coverings of $\mathbb{P}^1$-bundles over $S$  or as geometrically reducible conic bundles over $S$. 
\begin{thm}\label{Thm:fibralGeom}
The fibral divisors $D_0$, $D_1$, and $D_2$ are $\mathbb{P}^1$-bundles. 
$D_3$ and $D_4$ are double covers of $\mathbb{P}^1$-bundles branched at $V(a_{6,4})$.
The corresponding projective bundles are\footnote{We do not write explicitly the obvious pullback of line bundles. For example, if $\sigma: S\hookrightarrow B$ is the embedding of $S$ in $B$ and $\mathscr{L}$ is a line bundle on $B$, we abuse notation by writing
$ \mathbb{P}_S[\mathscr{L}\oplus \mathscr{O}_S]$ for  $\mathbb{P}_S[\sigma^* \mathscr{L}\oplus \mathscr{O}_S]$.}(see Figure \ref{Fig:FibralDiv})  

  $$
\begin{aligned}
\bullet \  D_0 & \quad\text{is isomorphic to } \quad \mathbb{P}_S[\mathscr{L}\oplus \mathscr{O}_S]\\
\bullet \  D_1  &\quad\text{is isomorphic to } \quad \mathbb{P}_S[\mathscr{L}^{\otimes 2}\oplus \mathscr{S}] \\
\bullet \  D_2  & \quad\text{is isomorphic to } \quad \mathbb{P}_S[\mathscr{L}^{\otimes 3}\oplus \mathscr{S}^{\otimes 2}] \\
\bullet \  {D}_3&\quad\text{is isomorphic to } \   \    \text{ a double covering of }  \quad  \mathbb{P}_S[\mathscr{L}^{\otimes 4}\oplus \mathscr{S}^{\otimes 3}]  \quad \text{ramified in $V(a_{6,4})\cap S$}\\
\bullet \  {D}_4 & \quad \text{is isomorphic to } \  \   \text{ a double covering of  }  \quad  \mathbb{P}_S[\mathscr{L}^{\otimes 2}\oplus \mathscr{S}^{\otimes 2}]\quad \text{ramified in $V(a_{6,4})\cap S$}
\end{aligned}
$$
where $\mathscr{L}$ is the fundamental line bundle of the Weierstrass model and $S$ is the zero scheme of a regular section of the line bundle $\mathscr{S}=\mathscr{O}_B(S)$. 
\end{thm}

\begin{proof}
The strategy for this proof is as follows.   We use the knowledge of the explicit sequence of blowups to parametrize each curve. 
Since each blowup has a center that is a complete intersection with normal crossing, each successive blowup gives a  projective bundles. The successive blowups give a tower of projective bundles over projective bundles. 
We keep track of the projective coordinates of each projective bundle relative to its base. 
 An important part of the proof is to properly normalize the relative projective coordinates when working in a given patch, as they are twisted with respect to previous blowups. 
 We show that $D_0$, $D_1$, and $D_2$ are $\mathbb{P}^1$-bundles over $S$ while $D_3$ and $D_4$ are  conic bundles defined by a double cover of a $\mathbb{P}^1$-bundle over $S$.

The fiber  $C_0$ can be studied after the first blowup since the remaining  blowups are away from $C_0$. We can work in the patch $x\neq 0$. We use the defining equation of $C_0$ to solve for $e_1$ since $x$ is a unit. 
We then observe that $C_0$ has the  parametrization 
$$
C_0\quad [t^2:t^3:1] [1:t:0], \quad t=y/x. 
$$
This is the usual normalization of a cuspidal cubic curve. It follows that $C_0$ is a rational curve parametrized by $t$. Since $t=y/x$ is a section of $\mathscr{L}$, it follows that the fibral divisor  $D_0$ is isomorphic to the  $\mathbb{P}^1$-bundle $\mathbb{P}_S[\mathscr{L}\oplus \mathscr{O}_S]$ over $S$. 

$D_1$ is the Cartier divisor $V(e_1)$ in $Y$, which corresponds to the complete intersection $V(e_1,e_2)$ in $X_4$. The generic fiber of $D_1$ over $S$ is the rational curve $C_1$, which is parametrized as
\begin{align}
C_1\quad  [0:0 :  z] 		[\ell_1 ( e_3 e_4 x) :0 : \ell_1 s] 
		 [\ell_1 \ell_2 y:0]
		[\ell_1\ell_3 x:0]
		[0:\ell_4\ell_3^{-1} e_3].\nonumber
	\end{align}
We can use $\ell_4$, $\ell_3$, and $\ell_2$ to fix the scalings. But  $C_1$ is parametrized by $[x:s]$ and $D_1$ is isomorphic to the $\mathbb{P}^1$-bundle 
$\mathbb{P}_S[\mathscr{L}^{\otimes 2}\oplus \mathscr{S}]$ over $S$. 

 $C_2$ is defined as the generic fiber with $e_2=e_4=0$. This gives 
\begin{align} C_2\quad
		 [0:0 :  z]  [0:0 : \ell_1 s] 
		 [\ell_1 \ell_2 y:\ell_1^{-1}\ell_2 e_1]
		[\ell_1\ell_3 x:0]
		[0 :\ell_4 \ell_3^{-1}e_3].\nonumber
	\end{align}
	Fixing the scaling as $\ell_1=s^{-1}$,\quad $\ell_3=s x^{-1}$, we see that $C_2$ is  a rational curve parametrized by $[y:s^2]$ and $D_2$ is isomorphic to $\mathbb{P}_S[\mathscr{L}^{\otimes 3}\oplus \mathscr{S}^{\otimes 2}]$.

 For $C_3$ take $\ell_1=s^{-1}$, $\ell_3=sx^{-1}$,    $\ell_2=s^{-1}$, $\ell_3=s x^{-1}$, 
	$$
	\begin{aligned}
	C_3	\quad&[0 : 0:  1]		\quad[0 : 0: 1] 		\quad
		 [y s^{-2}:1]		\quad
				\quad [1:0]  		\quad	[ \ell_4 s^2 x^{-1}:\ell_4 x s^{-1}].\nonumber
	\end{aligned}
	$$
The double cover is $\left(\frac{y}{s^{2}}\right)^2=a_{6,4}$. This is clearly a double cover of $D_3^+$, where $D_3^+$ is  
	$\mathbb{P}^1$-bundle over $S$, whose fiber is  parametrized by $[ s^2 x^{-1}:x s^{-1}]$. Such a $\mathbb{P}^1$-bundle is isomorphic to $\mathbb{P}_S[\mathscr{L}^{\otimes 4}\oplus \mathscr{S}^{\otimes 3}]$.
	\item For $C_4$, take $\ell_1=s^{-1}$, $\ell_2=s^{-1}$, $\ell_4=s\ell_3^{-1}$, 
	\begin{align}
	C_4 \quad	&[0 : 0:  1]		\quad[0 : 0: 1] 		\quad
		 [y s^{-2}:1]		\quad
				\quad	[ s^{-1} x: s] \quad [1:0]\nonumber
	\end{align}
	The double cover is again $\left(\frac{y}{s^{2}}\right)^2=a_{6,4}$.
\end{proof} 

\begin{thm}\label{Thm:degeneration}
The crepant resolution defined in Theorem  \ref{Thm:blowups}
has the following properties: 
\begin{enumerate}[label=(\roman*)]  \itemsep2pt
\item The resolved variety is a flat elliptic fibration over the base $B$. 
\item The fiber over the generic point of $S$  has dual graph $\widetilde{\text{F}}_4^t$ and the geometric generic fiber is of Kodaira type IV$^*$. 
\item The fiber degenerates over $V(s,a_{6,4})$ as 
$$
V(s,a_{6,4}) \quad 
	\begin{cases}
	C_3 \longrightarrow 2 C_3',\\
	C_4  \longrightarrow 2 C_4'.
	\end{cases}
	$$
 where $C_3$ and $C_4$ are generic curves defined over $S$, and $C_3'$ and $C_4'$ are generic curves over $V(s, a_{6,4})$.
The generic fiber over $V(s, a_{6,4})$ is a non-Kodaira fiber composed of five geometrically irreducible  rational curves. The reduced curves  meet  transversally  with multiplicities $1-2-3-4-2$.

	\end{enumerate}
\end{thm}

\begin{proof}
The special fiber is the fiber over the generic point of $S$. Note that $C_2$ and $C_3$ intersect at a divisor of degree two, composed of two points that are non-split. Hence, the dual graph of this fiber is the twisted affine Dynkin diagram of type $\widetilde{\text{F}}_4^t$. All the curves are geometrically irreducible with the exception of $C_3$ and $C_4$,  which are the double covers of a geometrically irreducible rational curve and the branching locus is $a_{6,4}=0$. 
Each of these two curves splits into two geometrically  irreducible curves in a field extension that includes a square root of  $a_{6,4}$. They  degenerate into a double  rational curve over $a_{6,4}=0$. Over the branching locus, the singular fiber is a chain $1-2-3-4-2$. 
The geometric generic fiber has a dual graph that is a $\widetilde{\text{E}}_6$ affine Dynkin diagram. 
The fibers $C_a$ are fibers of fibral divisors $D_a$. 
 The matrix of intersection numbers $\mathrm{deg}  (D_a\cdot C_a)$ is the opposite of the invariant form of the twisted affine Dynkin diagram of type $\tilde{\text{F}}_4^t$, normalized in such a way that the short  roots have length square $2$:
$$
	\textnormal{deg} (D_a \cdot C_b)= \left({
	\begin{array}{c|cccc}
	-2 & 1 & 0 & 0 & 0 \\
	\hline
	1 & -2 &  1 & 0 & 0 \\
	0 & 1 & -2 & 2 & 0 \\
	0 & 0 & 2 & -4 & 2 \\
	0 & 0 &  0 &2 & -4 
	\end{array}}\right)
$$	
This matrix has a kernel generated by the vector $(1,2,3,2,1)$. 
The entries of this vector give the multiplicities of the curve $C_a$, or equivalently, of the fibral divisors $D_a$. 

\end{proof}

\begin{exmp}
If  $B$ is the total space of the line bundle $\mathscr{O}_{\mathbb{P}^1}(-n)$ with $n\in \mathbb{Z}_{\geq 0}$, the Picard group of $B$ is generated by one element, which we call $\mathscr{O}(1)$.  
 In particular, the compact curve $\mathbb{P}^1$ is a section of $\mathscr{O}(-n)$. 
A local Calabi-Yau threefold can be defined by a Weierstrass model with 
 $\mathscr{L}=\mathscr{O}(2-n)$. 
 Consider the case of the F$_4$-model, defined with $S$ a regular section of $\mathscr{S}=\mathscr{O}(-n)$.  
 This requires that $1\leq n\leq 5$.  Denoting the Hirzebruch surface of degree $d$ by $\mathbb{F}_d$, we have 
$D_0=\mathbb{F}_{n-2}$, $D_1=\mathbb{F}_{n-4}$, $D_2=\mathbb{F}_{n-6}$, $D_3$ a double cover of $\mathbb{F}_{n-8}$, and $D_4$  a double cover of $\mathbb{F}_{4}$.
In particular for $n=5$, the divisors are  $D_0=\mathbb{F}_{3}$, $D_1=\mathbb{F}_{1}$, $D_2=\mathbb{F}_{1}$, $D_3$  a double cover of $\mathbb{F}_{3}$, and $D_4$  a double cover of $\mathbb{F}_{4}$. 
\end{exmp}

		\subsection{ Representation associated to the elliptic fibration}\label{Sec:Rep}
\begin{figure}[htb]
\begin{center}
\scalebox{1}{\begin{tikzpicture}
				\node[draw,circle,thick,scale=1.25,fill=black,label=below:{1}] (0) at (0,0){};
				\node[draw,circle,thick,scale=1.25,label=below:{2}] (1) at (.8,0){};
				\node[draw,circle,thick,scale=1.25,label=below:{3}] (2) at (.8*2,0){};
				\node[draw,circle,thick,scale=1.25] (3) at (.8*3.2,-1){};
				\node[draw,circle,thick,scale=1.25] (4) at (.8*4.6,-1){};
								\node[draw,circle,thick,scale=1.25] (5) at (.8*3.2,1){};
																\node[draw,circle,thick,scale=1.25] (6) at (.8*4.6,1){};
			        \node at (.8*3.15,0)[draw,dashed, line width=2pt, ellipse, minimum width=120pt, minimum height=22pt,rotate=90,yshift=-1pt]{};
				\node at (.8*4.55,0)[draw,dashed, line width=2pt, ellipse, minimum width=120pt, minimum height=22pt,rotate=90,yshift=-1pt]{};
									\node   at (.8*3.15,-2.5) { 2};\node   at (.8*4.55,-2.5) { 1};

				\draw[thick] (0)--(1)--(2)--(3)--(4);
				\draw[thick] (2)--(5)--(6);
				\draw[->,>=stealth',thick=4mm]  (5,0) -- (7,0);
				\node[draw,circle,thick,scale=1.25,fill=black,label=below:{1}] (0a) at (8,0){};
				\node[draw,circle,thick,scale=1.25,label=below:{2}] (1a) at (8+.8,0){};
				\node[draw,circle,thick,scale=1.25,label=below:{3}] (2a) at (8+.8*2,0){};
				\node[draw,circle,thick,scale=1.25,label=below:{4}] (3a) at (8+.8*3,0){};
				\node[draw,circle,thick,scale=1.25,label=below:{2}] (4a) at (8+.8*4,0){};
				\draw[thick] (0a)--(1a)--(2a)--(3a)--(4a);
			\node   at (6,.5) { \scalebox{1.25}{$\displaystyle a_{6,4}=0$}};
				\node (w1) at (9+.8*4,2) { \scalebox{1}{$\boxed{ 0\ 0\  1 \  -2}$}};
				\node (w2)  at (7+.8*4,2) { \scalebox{1}{$\boxed{0\ 1\  -2 \  1 }$}};
			\draw[->,>=stealth',thick=1mm,dashed]  (w1)--($(4a)+(0.2,.5)$);
								\draw[->,>=stealth',thick=1mmm,dashed]  (w2)--($(3a)+(-0,.5)$);

	\end{tikzpicture}}
\end{center}
\caption{Degeneration of the F$_4$ fiber at the  non-transverse collision  $\text{IV}^{*\text{ns}}+\text{I}_1$.
The nodes represent geometrically irreducible curves. The dashed lines identify the irreducible components of the generic fiber that are geometrically irreducible. 
They split inside their interior nodes after  a $\mathbb{Z}/2\mathbb{Z}$ field extension. 
The degeneration produces weights $\boxed{0\ 1\  -2 \  1 }$ and $\boxed{0\ 0\  1 \  -2 }$ that identify $\mathbf{26}$ as the representation associated with the elliptic fibration. 
This fiber can be seen as an incomplete Kodaira fiber of type III$^*$ with its dual graph $\tilde{\text{E}}_7$ if $\alpha=0$ or 
an incomplete Kodaira fiber of type II$^*$ with dual graph $\tilde{\text{E}}_8$ if $\alpha>0$. 
\label{Figure:Degeneration}}
\end{figure}
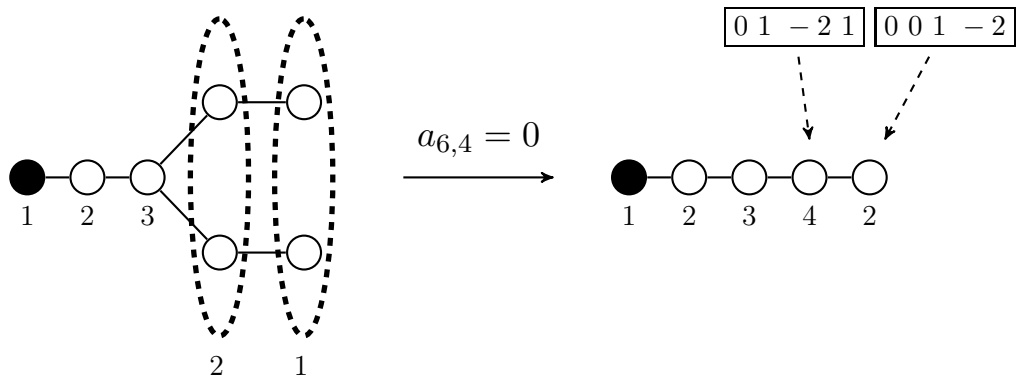

In this subsection, we compute the weights of the vertical curves appearing over codimension-two points. There is only one case to consider. 
The generic fiber over  $V(a_{6,4})\cap S$ is a fiber of type $1-2-3-4-2$ resulting from the following specialization:
	$$
	\begin{cases}
	C_3 \longrightarrow 2 C_3',\\
	C_4  \longrightarrow 2 C_4'.
	\end{cases}
	$$
	\begin{thm}\label{Thm:weights}
	\begin{enumerate}
\item 	The intersection numbers of the generic curves $C_3'$ and $C_4'$ with the fibral divisors D$_a$ for $a=0,1,\ldots, 4$   are 
	$$
\varpi(C_3')=(0,0,1,-2,1), \quad \varpi ( C_4')= (0,0,0,1,-2).
	$$
	\item The representation associated to an F$_4$-model  is  the quasi-minuscule representation $\mathbf{26}$ of F$_4$. 
	\end{enumerate}
	\label{thm:rep26}
	\end{thm}
		\begin{proof}
	
By the linearity of the intersection product,  
 the geometric weights  $\varphi(C)=(D_a\cdot C)$ of $C=C_3'$ and $C=C_4'$ are half of the geometric weights of $C_3$ and $C_4'$: 
	$$
	\varpi(C_3')=(0,0,1,-2,1), \quad \varpi ( C_4')= (0,0,0,1,-2).
		$$
	Ignoring the weight of $D_0$, we get the following two weights of F$_4$:
	$$
	\boxed{0\  1 \ -2 \ 1}\quad \boxed{0\ 0\  1 \  -2 }. 
	$$
	These two weights are quasi-minuscule and are in the same Weyl orbit, which consists of the non-zero weights of the representation $\mathbf{26}$ of F$_4$. See Theorem \ref{thm:rep26} for the proof. This is a  fundamental  representation corresponding to the fundamental weight $\alpha_4$. 
	This representation is also quasi-minuscule. 
	Since the weight system is invariant  under a change of signs,  the representation is quaternionic  and we can consider half of the representation. 
	Both weights coming from the degeneration of the main fiber are in the same half quaternionic set of weights.
	\end{proof}

		\begin{rem}
An important consequence of   Theorem \ref{thm:rep26} is that the elliptic fibration does not have flop transitions to another smooth elliptic fibration  since all the curves move in families. 
		\end{rem}

The reduced discriminant has two components, namely $S=V(s)$ and $\Delta'=V( a^3_{ 4,3+\alpha} s^{1+3\alpha}+27 a_{6,4}^2)$  intersecting non-transversally at $V(s,a_{6,4})$. 
Their intersection is  exactly the locus over which the fiber IV$^{*\text{ns}}$ degenerates. 
The generic fiber over 
$\Delta'$ is of Kodaira fiber I$_1$. Hence, what we are witnessing is a collision of type $\text{IV}^{*\text{ns}}+\text{I}_1$, leading to an incomplete $\text{III}^*$ or an incomplete  II$^*$
$$
\text{IV}^{*\text{ns}}+\text{I}_1\longrightarrow\quad 1-2-3-4-2\quad (\text{incomplete}\  \   \text{III}^* \quad \text{or}\quad \text{incomplete}\  \    \text{II}^*).
$$
This is clearly not a collision of Miranda models since the fibers have different $j$-invariants and do not intersect transversally. The $j$-invariants of  fibers of type IV$^*$  and I$_1$ are, respectively, zero and  infinity. 

By using an elliptic surface whose bases pass through the  collision point, the singular fiber at the collision point is  of Kodaira type III$^*$ for $\alpha=0$ and Kodaira type II$^*$ for $\alpha>0$.
Interestingly, we can think of the singular fiber $1-2-3-4-2$ as a contraction of a fiber of type III$^*$ or a fiber of type II$^*$,  as expected from the analysis of Cattaneo 
\cite{Cattaneo:2013vda}.

	\section{Topological invariants}\label{Sec:Topology}
	In this section we compute several topological invariants of the crepant resolution. 
Using the pushforward theorem of \cite{MMP1}, we can compute the Euler characteristic of an F$_4$-model. 
We need to know the classes of the centers of the sequence of blowups  that define the crepant resolution. 
The center of the $n$th blowup is a smooth complete intersection of $d_n$ divisors of classes $Z_i^{(n)}$, where $i=1,2,\cdots, d_n$.

We recall that $H=c_1(\mathscr{O}(1))$, $L=c_1(\mathscr{L})$, and $S=[S]$.  	
 The classes associated to the  centers of each blowup are \cite{MMP1}	
	\begin{align}
	\begin{split}
	\begin{array}{lll}
		Z^{(1)}_1 = H+2 L\quad\quad\quad & Z_2^{(1)} =H + 3 L\quad\quad\quad&Z_3^{(1)} = S\\
		Z^{(2)}_1 =Z^{(1)}_2 - E_1  & Z^{(2)}_2 = E_1 & \\
			Z^{(3)}_1 =Z^{(1)}_1 -E_1& Z^{(3)}_2 =E_2& \\
				Z^{(4)}_1 = E_2 - E_3 & Z^{(4)}_2 =E_3 & 
				\end{array}
	\end{split}
	\end{align}
\begin{thm}[\cite{MMP1}]
	 The  Euler characteristic of an F$_4$-model obtained by a crepant resolution of the Weierstrass model 
$ y^2z=x^3+ s^{3+\alpha}a_{4,3+\alpha} x z^2 + s^4 a_{6,4} z^3\quad (\alpha\in\mathbb{Z}_{\geq 0})$ over the base $B$ 
 is
	\begin{align}\nonumber
\chi(Y)=\int	 12 \frac{(  L + 3 S L-2 S^2)}{(1+ S)(1  + 6 L-4 S ) } c(B),
	\end{align}
	where $L=c_1(\mathscr{L})$ and $S$ is the class of $V(s)$. 
	In particular, denoting $c_i(TB)$ simply as $c_i$:  
	\begin{center}
\begin{tabular}{|c|c|}
\hline 
  & $\chi(Y)$  \\
 \hline 
$\dim Y=3$& $12  (c_1 L-6 L^2+6 L S-2 S^2)$ \\
$\dim Y=3$ and $c_1(TY)=0$ &$12( 5 c_1^2-6 c_1 S+2 S^2)$\\
\hline
$\dim Y=4$ &$12( -6 c_1 L^2+c_2 L+36 L^3+6 c_1 L S-60 L^2 S-2 c_1 S^2+34 L S^2-6 S^3)$ \\
$\dim Y=4$ and $c_1(TY)$ & $12 t^3 (30 c_1^3  + c_1 c_2  - 54 c_1^2 S+32 c_1 S^2 - 6 S^3)$\\
\hline 
\end{tabular}
\end{center}
\end{thm}
	\begin{thm}[\cite{MMP1}]
Let $Y$  be a Calabi-Yau threefold that is an F$_4$-model obtained from a crepant resolution. Then  the Hodge numbers of $Y$ are
	\begin{align}
		h^{1,1}(Y) &= 15 - K^2,\quad h^{2,1}(Y)= 15  + 29 K^2 + 36 S K + 12 S^2.\nonumber 
	\end{align}
\end{thm}

With an explicit resolution of singularities, it is straightforward to compute intersection numbers of divisors. 
In particular, we evaluate the triple intersection numbers of the fibral divisors $D_a$, where $a=0,1,2,3,4$. 
The result is
\begin{equation}\varphi_* \left(\left(\sum D_a \phi_a \right)^3\cdot \varphi^* M\right)=6\mathscr{F}(L,S,\phi) M,\nonumber
\end{equation}
where $M$ is an arbitrary element of $A_{d-2}(B)$. In particular, if the base $B$ is a surface, $M$ is just a point. 
\begin{thm}
Let $D_a$ ($a=0,1,2,3,4$) be the fibral divisor of an F$_4$-model obtained by a crepant resolution of singularities. 
 The triple intersection numbers $\varphi_* \Big((\sum D_a \phi_a )^3\cdot \varphi^* M\Big)=6\mathscr{F}(L,S,\phi) M$, where $M$ is an element of $A_{d-2}(B)$ ($d=\dim\ B$), are given by
\begin{align}
6 \mathscr{F}(L,S,\phi)=& 4 (L-S) S\, \phi _0^3 +3 (S-2 L) S\, \phi _0^2 \phi _1  +3 L S \phi _0 \phi _1^2 \nonumber  \\
&+ 4  (L-S) S\,  \phi _1^3+4 (L-S) S\,  \phi _2^3 +8 (S-2 L)S\,   \phi _3^3+8  (S-2 L) S\, \phi _4^3\nonumber  \\
& +3 (2 S-3 L)S\,  \phi _1^2 \phi _2   +3  (2 L-S) S\,  \phi _1 \phi _2^2\nonumber  \\
&+6 (3 S-4 L)S\,  \phi _2^2 \phi _3 +12 (3 L-2 S)S\,  \phi _2 \phi _3^2+12 (S-L)S\,  \phi _3^2 \phi _4
+6 (4 L-3 S)S\,  \phi _3 \phi _4^2.\nonumber
\end{align}
\end{thm}
\begin{proof}
 Use equation \eqref{eqn:divisorClass} and  successively apply the  pushforward formula of \cite{MMP1}.
\end{proof}
In the case of a Calabi-Yau threefold, we have $L=-K_B$. It follows that we can express the coefficient in terms of the genus of $S$ and its self-intersection using the relation $2-2g=-K_B\cdot S -S^2$:
\begin{align}\nonumber
6\mathscr{F}=& - 8 \left(g-1\right) \phi _0^3-3   \left(-4 g+4+S^2\right)\phi _0^2\phi _1+3  \left(-2 g+2+S^2\right)\phi _1^2 \phi _0
\\
&-8 (g-1) \phi _1^3 -8 (g-1) \phi _2^3+8  \left(4 g-4-S^2\right)  \phi _3^3+
8\left(4 g-4-S^2\right) \phi _4^3 \nonumber\\
&3  \left(6 g-6-S^2\right) \phi _1^2 \phi _2 +3 \left(-4 g+4+S^2\right) \phi _1  \phi _2^2\nonumber\\
&+6 \left(8 g-8-S^2\right) \phi _2^2 \phi _3
+12 \left(-6 g+6+S^2\right)  \phi _2\phi _3^2 +24 (g-1) \phi _3^2 \phi _4 +6\left(-8 g+8+S^2\right)  \phi _3 \phi _4^2.\nonumber
\end{align}

In the case of a threefold, the second Chern class defines a linear form on $H^2(Y, \mathbb{Z})$. In particular, for the fibral divisors we have
\begin{align}\nonumber
& \int_Y c_2(TY)\cdot (\sum_a D_a \phi_a)=
2 S (S-L) (\phi _0+\phi _1+\phi_2) +4S(2L-S)( \phi _3+ \phi _4).
\end{align}
Imposing the Calabi-Yau condition, we can rewrite this as
\begin{align}\nonumber
& \int_Y c_2(TY)\cdot (\sum_a D_a \phi_a)=
4(g-1) (\phi _0+\phi _1+\phi_2) +4(4-4g+S^2)( \phi _3+ \phi _4).
\end{align}

\subsection{Stein factorization and the geometry of non-simply laced G-models}

To understand the geometry of a fibral divisor $D$, it is important to see the divisor $D$ as  a relative scheme with respect to the appropriate base. 
The choice of the base is crucial to having the correct physical interpretation.  
In particular, to discuss the matter content of the theory, the base has to be a component of the discriminant locus. 
However, to study the possible contractions of $D$, the base can be an arbitrary subvariety of the elliptic fibration. 

Let $S$ be the irreducible  component of the discriminant locus supporting the gauge group. 
In the case of $G$-models with $G$ a  non-simply laced groups, Stein factorization illuminates the discussion of the geometry of the fibral divisors $D$, whose generic fibers over $S$ are not geometrically irreducible. 
The elliptic fibration $\varphi:Y\to B$ pulls back to a fibration $D\to S$. If the generic fiber of this fibration is not geometrically irreducible, the generic fiber is not geometrically connected.

We recall the following two classical theorems on morphisms that are consequences of the theorem of formal functions. 
\begin{thm}[Zariski]
Let $f: X\to Y$ be a proper morphism of Noetherian schemes such that $f_* \mathscr{O}_X\cong \mathscr{O}_Y$. Then all fibers are geometrically connected and non-empty.  
\end{thm}

\begin{thm}[Stein factorization{ \cite[Chap III.11.3]{Hartshorne}}] 
Let $f: X\to S$ be a proper morphism with $S$ a Noetherian scheme. Then there exists a factorization 
 $$\begin{tikzpicture}
	\node(X) at (-2,0){$X$};
	\node(Y1) at (2,0){$S'$};
	\node(Y0) at (0,-2){$S$};
		\draw[big arrow] (X) -- node[below=.1cm,left]{$f$} (Y0);	
		\draw[big arrow] (X) -- node[above=.1]{$f'$} (Y1);	
				\draw[big arrow] (Y1) -- node[below=.2cm,right]{$\pi$} (Y0);
	\end{tikzpicture}
	$$
such that 
\begin{enumerate}
\item $\pi:S'\to S$ is a finite morphism and $f':X\to S'$ is a proper morphism with geometrically connected fibers.
\item  $f'_* \mathscr{O}_X\cong \mathscr{O}_{S'}$. 
\item $S'$ is the normalization of $S$ in $X$.  
\item $S'=\underline{\mathrm{Spec}}_S (f_* \mathscr{O}_X)$.
 \end{enumerate}
\end{thm}

The Stein factorization on $f: D\to S$ is the decomposition $f=\pi'\circ f'$, where 
$\pi':\overline{S}\to S$ is a finite map of degree $d$ and $f':D\to \overline{S}$ is a morphism with connected fibers. 
We expect  $f':D\to \overline{S}$ to be a $\mathbb{P}^1$-bundle and $\pi': \overline{S}\to S$  a smooth $d$-cover of $S$.

In the case of F$_4$-models, $D$ is a $\mathbb{P}^1$-bundle for the fibral divisor $D_0$, $D_1$, and $D_2$  (corresponding to the affine root, and the small roots of $\widetilde{F}_4^t$. 
The remaining two fibral divisors (namely $D_3$ and $D_4$) are not $\mathbb{P}^1$-bundles over $S$ but rather double covers of $\mathbb{P}^1$ bundles over $S$ with a ramification locus $a_{6,4}=0$. The crepant resolution naturally defines a double cover $\pi: D\to \underline{D}$, where  $\underline{D}$ is a $\mathbb{P}^1$-bundle $p:\underline{D}\to S$.  
Let $f=\pi\circ p: D\to S$ be the composition. 
The key to understanding the different perspective on the geometry of $D$ is to consider the Stein factorization of $f$.

Let $D$ be the reduced fibral divisor $D_3$ or $D_4$ of an F$_4$-model. 
By definition, $f:D\to S$ has a generic fiber that is not geometrically connected.    
Consider the Stein factorization of the morphism $f: D\to S$. It gives a factorization  $f=\pi'\circ f'$ with $\pi'$ a finite map and $f'$ a proper morphism with geometrically connected fibers.

 $$\begin{tikzpicture}
	\node(X) at (-2,0){$D$};
	\node(Y) at (0,0){$\underline{D}$};
	\node(S0) at (0,-2){$S$};
	\node(S1) at (-2,-2){$\overline{S}$};
	\draw[big arrow] (X) -- node[above,midway]{$\pi$} (Y);	
		\draw[big arrow] (Y) -- node[above,right]{$p$} (S0);	
				\draw[big arrow] (X) -- node[below=.2cm,left]{$f$} (S0);
				\draw[big arrow] (X) -- node[above,left]{$f'$} (S1);	
				\draw[big arrow] (S1) -- node[below=.1cm]{$\pi'$} (S0);	
	\end{tikzpicture}
	$$
	In particular, the morphism  $f':D\to \overline{S}$ endows $D$ the structure of a $\mathbb{P}^1$-bundle over the double cover  $\overline{S}$ of $S$. 
This structure illustrates that $D$ can contract to $\overline{S}$. 

It is important to not confuse the role of the morphisms $f':D\to \overline{S}$  and $f:D\to S$ in F-theory. One might naively assume that the existence of a $\mathbb{P}^1$-bundle $f':D\to \overline{S}$ means that the divisor $D$ does not produce new curves leading to localized matter representations. 
However, it is important to keep in mind that   it is  the morphism $f:D\to S$ over the curve $S$ that is relevant for studying  the singular fibers of the elliptic fibration as $S$. 

The morphism $f:D\to S$ contains singular fibers that are double lines. The intersection numbers of these lines with the fibral divisors give two weights 
of the representation $\mathbf{26}$,  namely $\boxed{0\  1 \ -2 \ 1}$ for $D_3$ and $\boxed{0\ 0\  1 \  -2 }$ for $D_4$.  
The same weights are obtained over any closed points away from $a_{6,4}$ and are attributed to non-localized matter.

\section{Application to M-theory and F-theory in 5 and 6 dimensions}\label{Sec:Mtheory}

In this section, we study the aspects of five-dimensional gauge theories with gauge group F$_4$ using the geometry of the F$_4$-model.
 We consider an M-theory compactified on an F$_4$-model $\varphi: Y\to B$.
  We assume then that the variety $Y$ is a Calabi-Yau threefold and the base $B$ is a rational surface. Then, the resulting theory is a five-dimensional $\mathscr{N}=1$ supersymmetric theory with eight supersymmetric generators, whose matter content contains $n_H$ hypermultiplets and $n_V$ vector multiplets. We have  $n_H^0$ neutral hypermultiplets, $n_{\mathbf{52}}$ hypermultiplets transforming in the adjoint representation, and $n_{\mathbf{26}}$ hypermultiplets transforming in the fundamental representation $\mathbf{26}$. 
  We have $n_V$ vector multiplets whose kinetic terms and Chern-Simons terms are controlled by a cubic prepotential. For the F$_4$ gauge theory with  both adjoint and fundamental matters, there is a unique Coulomb phase. 

Since F$_4$  does not have a non-trivial third order Casimir, the classical part of the prepotential vanishes, and the quantum corrections fully determine the prepotential. The number of vector multiplets is the dimension of F$_4$. Then, we can determine the quantum contribution to the prepotential and hence determine $n_{\mathbf{52}}$ and  $n_{\mathbf{26}}$.  
Since we know the Hodge numbers of F$_4$-models on a Calabi-Yau threefold, we can compute $n_H^0=h^{2,1}(Y)-1$ as well.

We also check that the data we collected geometrically for the five-dimensional gauge theory will satisfy the anomaly cancellation conditions in the uplifted six-dimensional theory with the same gauge group F$_4$, the same matter contents, and an addition of  $n_T=h^{1,1}(B)-1$   tensor multiplets. 
\subsection{Intriligator-Morrison-Seiberg potential}
In this paper, the Intrilligator-Morrison-Seiberg (IMS) prepotential is the quantum contribution to the prepotential of a five-dimensional gauge theory with the matter fields in the representations $\mathbf{R}_i$ of the gauge group. Let $\phi$ be in the Cartan subalgebra of a Lie algebra $\mathfrak{g}$. 
The  weights are in the dual space of the Cartan subalgebra.  We denote the evaluation of a  weight on $\phi$ as a scalar product $\langle \mu,\phi \rangle$.  We recall that the roots are the weights of the adjoint representation of  $\mathfrak{g}$.
Denoting the fundamental roots by $\alpha$ and the weights of $\mathbf{R}_i$ by $\varpi$ we have 
\begin{align}
6\mathscr{F}_{IMS} =&\frac{1}{2} \left(
{
\sum_{\alpha} |\langle \alpha, \phi \rangle|^3-\sum_{\mathbf{R}_i} \sum_{\varpi\in W_i} n_{\mathbf{R}_i} |\langle \varpi, \phi\rangle|^3 
}
\right).
\end{align}
For all simple groups with the exception of SU$(N)$ with $N\geq 3$, this is the full cubic prepotential as there are non-trivial third Casimir invariants. 

One complication to the formula is dealing with the absolute values. For a given choice of a group $G$ and representations $\mathbf{R}_i$, we have to determine a Weyl chamber to remove the absolute values in the sum over the roots. 
 We then consider the arrangement of hyperplanes $\langle \varpi, \phi\rangle=0$, where $\varphi$ runs through all the weights of all the representations $\mathbf{R}_i$. 
 If none of these hyperplanes intersect the interior of the Weyl chamber, we can safely remove the absolute values in the sum over the weights. 
 Otherwise, we have hyperplanes partitioning the fundamental Weyl chamber into subchambers. Each of these subchambers is defined by the signs of the linear forms $\langle \varpi, \phi\rangle$. 
 Two such subchambers are adjacent when they differ by the sign of a unique linear form. 
 
 Within each of these subchambers, the prepotential is a cubic polynomial; in particular, it has smooth second derivatives. 
 But as we go from one subchamber to an adjacent one, we have to go through one of the walls defined by the weights and the second derivative will not be well-defined. 
 Physically, we think of the Weyl chamber as the ambient space and each of the subchambers is called a Coulomb phase of the gauge theory. 
 The transition from one chamber to an adjacent chamber is a phase transition that geometrically corresponds to a flop between different crepant resolutions of the same singular Weierstrass model.

The number of chambers of such a hyperplane arrangement is physically the number of phases of the Coulomb branch of the gauge theory. 
In the case of F$_4$ with the matter fields in the representation $\mathbf{52}\oplus\mathbf{26}$, there is a unique chamber as the hyperplanes $\langle \varpi, \phi\rangle=0$  have no intersections with the interior of the fundamental Weyl chamber. 
The explicit computation of $6\mathscr{F}_{IMS}$ is presented in the theorem below.

\begin{thm} The  prepotential for a gauge theory  with the gauge group F$_4$ coupled to $n_{\mathbf{52}}$  hypermultiplets  in the adjoint representation and  $n_{\mathbf{26}}$  hypermultiplets  in the fundamental representation is 
\begin{align}
6\mathscr{F}_{IMS} =&
-8 \left(n_{\mathbf{52}}-1\right)\phi _1^3 
-8 \left(n_{\mathbf{52}}-1\right)  \phi _2^3-8  \left(n_{\mathbf{52}}+n_{{\mathbf{26}}}-1\right) \phi _3^3-8  \left(n_{\mathbf{52}}+n_{{\mathbf{26}}}-1\right)\phi _4^3 \nonumber  \\
& -3  \left(-n_{\mathbf{52}}+n_{{\mathbf{26}}}+1\right)\phi _1^2\phi _2 +3 \left(n_{\mathbf{52}}+n_{{\mathbf{26}}}-1\right)  \phi _1\phi _2^2\\
&+12 \left(-n_{\mathbf{52}}+n_{{\mathbf{26}}}+1\right) \phi _2 \phi _3^2-6  \left(-3 n_{\mathbf{52}}+n_{{\mathbf{26}}}+3\right)\phi _2^2 \phi _3 \nonumber  \\
&
+6  \left(-3 n_{\mathbf{52}}+n_{{\mathbf{26}}}+3\right)\phi _3 \phi _4^2 +24\left(n_{\mathbf{52}}-1\right) \phi _3^2 \phi _4. \nonumber
\end{align}
\end{thm}

\begin{thm}\label{Thm:n52andn26}
The triple intersection polynomial  of the elliptic fibration defined by the  crepant resolution of the F$_4$-model Weierstrass  matches the IMS potential if and only if 
\begin{equation}
n_{\mathbf{52}}=g, \quad n_{\mathbf{26}}=5 (1-g) + S^2.
\end{equation}
\end{thm}
\begin{proof} 
These numbers are obtained by comparing the coefficients of $\phi_1^3$ and $\phi_4^3$ in $\mathscr{F}_{IMS}$ and $\mathscr{F}$. A direct check shows that all the other coefficients match. 
\end{proof}

\subsection{Six-dimensional uplift and anomaly cancellation conditions}
In this section, we review the basics of anomaly cancellation in six-dimensional supergravity with the idea of applying it to the matter content we have identified in the previous section. 

Consider an ${\cal N}=1$  six-dimensional theory coupled to $n_T$ tensor multiplets, $n_V$ vectors, and $n_H$ hypermultiplets. 
The pure gravitational anomaly (proportional to $\mathrm{tr}\  R^4$) is canceled by the vanishing of its coefficient: 
\begin{subequations}
\begin{equation}\label{Eq.grav1}
n_H-n_V+29n_T-273=0.
\end{equation}
If the gauge group is a simple group $G$, 
\begin{align}\label{Eq.grav2}
n_H=n_H^{ch}+n_H^0, \quad n^{ch}_H =\sum_{\mathbf R} n_{\bf R} \left( \rm{dim}{\mathbf R} -\rm{dim} {\mathbf R}_0 \right), \quad n_H^0=h^{2,1}(Y)+1, \quad n_V=\dim  G,
\end{align}
\end{subequations}
where $\rm{dim} \mathbf{R}_0$ is the number of zero weights in $\mathbf{R}$ and $ \Big(\rm{dim}{\mathbf R} -\rm{dim} {\mathbf R}_0\Big)$ is the {\em charged dimension} of $\mathbf{R}$.
  Assuming the gravitational anomaly is canceled, the  anomaly polynomial is a function of the matter content of the theory. It depends on the representations and their multiplicities
 \begin{equation}
 \text{I}_8= \frac{9-n_T}{8} (\mathrm{tr} \   R^2)^2+\frac{1}{6}\  X^{(2)} \mathrm{tr}\ R^2-\frac{2}{3} X^{(4)},
 \end{equation}
where $X^{(n)}$ contains the information on the representations and their multiplicities and is given by
\begin{align}
X^{(n)}=\mathrm{tr}_{\mathrm{adj}}\  F^n -\sum_{\mathbf R}n_{\mathbf R} \mathrm{tr}_{\mathbf R}\  F^n.
\end{align}
 To simplify the expression of I$_8$, we choose a reference representation $\mathbf F$,  and introduce the  coefficients $A_{\bf R}$, $B_{\bf R}$, and $C_{\bf R}$ defined by the  trace identities
\begin{equation}
\mathrm{tr}_{\bf R}\  F^2 = A_{\bf R} \mathrm{tr}\  F^2 , \quad 
\mathrm{tr}_{\bf R}\  F^4 =  B_{\bf R} \mathrm{tr}\   F^4  +  C_{\bf R} (\mathrm{tr}\  F^2)^2.
\end{equation}
In a theory with only one quartic Casimir, we simply have $B_R=0$. 
It is very useful to  reformulate $X^{(n)}$ in terms of a unique representation $F$:
\begin{align}
X^{(2)}&=\Big(A_{\text{adj}}   -\sum_{\mathbf R}n_{\mathbf R} A_{\mathbf R}\Big) \mathrm{tr}_{\mathbf F}  F^2,\quad 
X^{(4)}=\Big(B_{\text{adj}}   -\sum_{\mathbf R}n_{\mathbf R} B_{\mathbf R}\Big) \mathrm{tr}_{\mathbf F}  F^4 +
\Big(C_{\mathrm{adj}}   -\sum_{\mathbf R}n_{\mathbf R} C_{\mathbf R}\Big)( \mathrm{tr}_{\mathbf F} F^2)^2
.
\end{align}
 To have a chance to cancel the anomaly in a theory with at least two quartic Casimirs, the coefficient of $\mathrm{tr}_{\mathbf F}  F^4$ must vanish.
This condition is irrelevant in a theory with only one quartic Casimir since $\mathrm{tr}_{\mathbf F}  F^4$ is proportional to $(\mathrm{tr}_{\mathbf F}  F^2)^2$.
If the following three conditions are satisfied, we can deduce that the anomalies are canceled by the Green-Schwartz mechanism:
\begin{itemize}
\item $n_H-n_V+29n_T-273=0$.
\item $ B_{\text{adj}}   -\sum_{\mathbf R}n_{\mathbf R} B_{\mathbf R}=0$ (in a theory with at least two quartic Casimirs).
\item I$_8$ factorizes.
\end{itemize}

\subsection{Cancellations of six-dimensional anomalies for an F$_4$-model}
 In this section, we prove that the data we computed on the F$_4$-model as seen in a M-theory compactification on a Calabi-Yau threefold $Y$ will satisfy the anomaly cancellation conditions of a six-dimensional theory with the same gauge group and same number of vector and hypermultiplets. 
Moreover, in a $(1,0)$ six-dimensional gauge theory, we can also have tensor multiplets. 
We assume here that the tensor multiplets are massless and have numbers $n_T=h^{1,1}(B)$, which is $n_T=9-K^2$ since we assume that the base is a rational surface. 

We recall the data we will need, which were computed in previous sections: 
\begin{align}\nonumber
		n_T=9-K^2, \quad
		h^{2,1}(Y)= 15  + 29 K^2 + 72(g-1) -24 S^2, \quad n_V=52, \quad 
n_{\mathbf{52}}=g, \quad n_{\mathbf{26}}=5 - 5 g + S^2.
	\end{align}
We recall that $ n_V$ is  given by the dimension of the Lie algebra of F$_4$. 
The numbers $n_{\mathbf{52}}$ and $n_{\mathbf{26}}$ were computed using the Intrilligator-Morrison-Seiberg prepotential.   The Hodge numbers of $Y$ were computed in \cite{MMP1}. 

We compute first the  pure gravitational anomaly. 
We need to satify equation \eqref{Eq.grav1}. Using the data of \eqref{Eq.grav2}, we  compute 
$n_H=29 K^2+64$ and check that the  gravitational anomaly cancels: 
\begin{equation}\nonumber
n_H -n_V +29 n_T -273=(29 K^2+64)-52+29(9-K^2)-273=0.
\end{equation}
We will now show that the anomaly polynomial  I$_8$ is a perfect square. 
Since F$_4$ does not have a fourth Casimir, $B_{\mathbf{R}}=0$. Taking $\mathbf{26}$ as our reference representation \cite{Avramis:2005hc} ,
$$
\mathrm{tr}_{\mathbf{52}}\  F^2=3 \mathrm{tr}_{\mathbf{26}}\  F^2, \quad 
\mathrm{tr}_{\mathbf{52}}\  F^4=\frac{5}{12} ( \mathrm{tr}_{\mathbf{26}}   F^2)^2\quad 
\mathrm{tr}_{\mathbf{26}}\  F^4=\frac{1}{12} ( \mathrm{tr}_{\mathbf{26}}   F^2)^2.
$$
From these trace identities, we can immediately read off the coefficients $A_{\mathbf{R}}$, $B_{\mathbf{R}}$, and $C_{\mathbf{R}}$:
\begin{align}\nonumber
A_{\mathbf{52}}=3, \quad 
B_{\mathbf{52}}=0, \quad 
 C_{\mathbf{52}}=\frac{5}{12}, \quad 
A_{\mathbf{26}}=1, \quad  
B_{\mathbf{26}}=0, \quad 
  C_{\mathbf{26}}=\frac{1}{12}.
\end{align}
Hence, the forms $X^{(2)}$ and $X^{(4)}$ are 
\begin{align}\nonumber
X^{(2)}=(3-3 n_{\mathbf{52}}-n_{\mathbf{26}}) \mathrm{tr}\  F^2, \quad X^{(4)}=  \frac{1}{12}\Big( 5-5 n_{\mathbf{52}}  - n_{\mathbf{26}}     \Big)( \mathrm{tr}   F^2)^2.
\end{align}
After plugging in the values of $n_{\mathbf{52}}$ and $n_{\mathbf{26}}$, we have 
\begin{align}\nonumber
 X^{(2)}=K\cdot S \mathrm{tr}\  F^2, \quad X^{(4)}=  -\frac{1}{12} S^2( \mathrm{tr}   F^2)^2.
\end{align}
 We can now prove that the anomaly polynomial is a perfect square:
\begin{align}\nonumber
\text{I}_8= \frac{K^2}{8} (\mathrm{tr} \   R^2)^2+\frac{1}{6} K S (\mathrm{tr}\  F^2)(\mathrm{tr} \   R^2)+\frac{1}{18}S^2( \mathrm{tr}   F^2)^2=\frac{1}{72}\Big({3  K  \mathrm{tr} \   R^2+2 S \    \mathrm{tr}   F^2  }\Big)^2.
\end{align}
This shows that the anomalies can be canceled by the traditional Green-Schwarz mechanism.

\subsection{Frozen representations}\label{sec:frozen}

Motivated by the counting of charged hypermultiplets in M-theory compactifications, we introduce the notion of {\em frozen representations}. 
When a vertical curve of an elliptic fibration carries the weight of a representation $\mathbf{R}$, the representation $\mathbf{R}$ is said to be the {\em geometric representation} induced by the vertical curve. 
The existence of a vertical curve carrying a weight of the representation $\mathbf{R}$  is a necessary but not sufficient condition for hypermultiplets to be charged under the representation $\mathbf{R}$. 
It is possible that a geometric representation is not physical in the sense that no hypermultiplet is charged under  $\mathbf{R}$, so that   $n_{\mathbf{R}}=0$. 

\begin{defn}[Frozen representation]\label{defn:frozen}
A geometric representation $\mathbf{R}$ is said to be {\em frozen} if it is induced by the weights of vertical curves of the elliptic fibration but no hypermultiplet is charged under $\mathbf{R}$. 
\end{defn}
Witten has proven that a curve of genus $g$ supporting a Lie group  produces $g$  hypermultiplets in the adjoint representation. 
It follows that the adjoint representation is frozen when the genus is zero. 
This is because the adjoint hypermultiplets  are counted by holomorphic forms on the curve $S$. 
When the genus is zero, there are no such forms. Hence, even though we clearly witness vertical rational curves carrying the weights of the adjoint representation, no adjoint hypermultiplet are to be seen.

For non-simply laced groups, frozen representations can occur when the curve defined by the Stein factorization has the same genus as  the base curve $S$.  The representation is frozen if and only if 
$g=0 \quad \text{and} \quad \textnormal{deg}\ R =  2 (d-1)$. 

The number of points over which the  fiber IV$^{*\text{ns}}$ degenerates to the non-Kodaira fiber IV$^*_2$ is the degree of the ramification divisor $R$. 
 This  is  an interesting geometric invariant to keep in mind.
\begin{lem}\label{Lem:degR}
  If $B$ is a surface and $\mathscr{L}=\mathscr{O}_B(-K_B)$ so that the F$_4$-model is a Calabi-Yau threefold, then the number of points over which the generic fiber IV$^{*\text{ns}}$ over $S$ degenerates  further to the non-Kodaira fiber $1-2-3-4-2$ is 
$$ \textnormal{deg}\  R=12(1-g)+2 S^2.$$
\end{lem}
\begin{proof}
The number of intersection points is the intersection product of $V(a_{6,4})$ and the curve $S$. Note that  this is a transverse intersection and the class of $V(a_{6,4})$ is $6L-4S$. Then using 
$L=-K_B$ and $2g-2=(K+S)\cdot S$, we have 
$
\textnormal{deg}\ R=(6L-4S)\cdot S
=12 (1-g)+2S^2.
$
 \end{proof}

The degree of the ramification locus $R$  is the number of points in the reduced intersection of the two components of the discriminant locus,  namely $S$ and $\Delta'$. 
The degree of $R$ has to be positive as otherwise $D_3$ and $D_4$ are not irreducible and we have  an E$_6$-model rather than an F$_4$-model.
This gives a constraint on the self-intersection of $S$: 
$S^2>6(g-1)$. 
For example, if $B=\mathbb{P}^2$,  the bound is respected when $S$ is a smooth curve of degree $1$, $2$, $3$, or $4$.

  \begin{prop}\label{Prop:frozen}
  for an F$_4$-model, the number of charged hypermultiplets are 
  $$n_{\mathbf{26}}=\frac{1}{2} \textnormal{deg}\ R+g-1\quad   n_{\mathbf{52}}=g.$$
 In particular, the adjoint representation is frozen if and only if $g=0$. The 
   representation $\mathbf{26}$ is frozen if and only if  $S^2=-5$ and $g=0$, which also forces the adjoint representation to be frozen. 
        \end{prop}
      \begin{proof}
       The number of representation $n_{\mathbf{26}}$  is computed in Theorem \ref{Thm:n52andn26}.  The representation $\mathbf{26}$ is frozen if and only if $\textnormal{deg}\ R=2(1-g)$. Since the degree of $R$ has to be positive, we also see that $g=0$, hence by using 
       Lemma \ref{Lem:degR}, we conclude that $S^2=-5$ and $\textnormal{deg}\ R=2$.
      \end{proof}

We consider two  important examples. 

\begin{exmp}[Frozen adjoint representation]
$n_{\mathbf{52}}=0$ if and only if  $g=0$.   In this case, $\textnormal{deg}\ R=12+2S^2$ and $n_{\mathbf{26}}=5+S^2$. 
This matches what is found in Table 3 of \cite{Bershadsky:1996nh}  using $n=S^2$ as the instanton number. 
\end{exmp}

 \begin{exmp}[Frozen representation  $\mathbf{26}$]
$n_{\mathbf{26}}=0$ if and only if  $S^2=-5$ and $g=0$. 
For example, take $B$  to be the quasi-projective surface given by  the total space of the line bundle $\mathscr{O}_{\mathbb{P}^1}(-5)$. 
To construct a local Calabi-Yau threefold, use $\mathscr{L}=\mathscr{O}_B(-3)$. This is the non-Higgsable model of \cite{Morrison:2012np}. The defining equation of such a Weierstrass model is 
$$y^2 z= x^3 + f_3 s^3 x z^2 + g_2 s^4 z^3,$$ where $g_2$ and $f_3$ are respectively sections of $\mathscr{O}_{\mathbb{P}^1}(2)$ and $\mathscr{O}_{\mathbb{P}^1}(3)$.
The IV$^{*\text{ns}}$ fiber degenerates further at the two points  $g_2=0$. Over these points we have the non-Kodaira fiber of type $1-2-3-4-2$, each carrying  the weights of  the representation $\mathbf{26}$. 
\end{exmp}

\subsection{Geometry of fibral divisors in the case of frozen representations}

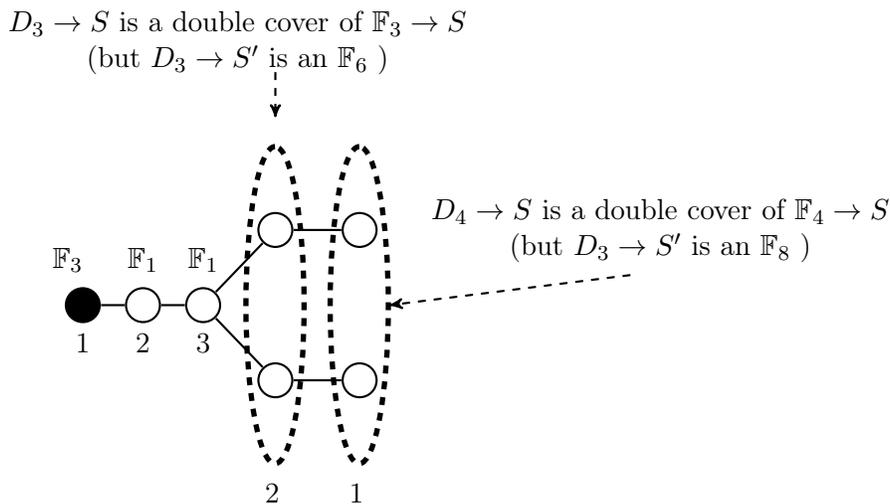
\begin{figure}[htb]
\begin{center}
\begin{tikzpicture}
				\node[draw,circle,thick,scale=1.25,fill=black,label=below:{1}] (0) at (0,0){};
				\node[draw,circle,thick,scale=1.25,label=below:{2}] (1) at (.8,0){};
				\node[draw,circle,thick,scale=1.25,label=below:{3}] (2) at (.8*2,0){};
				\node[draw,circle,thick,scale=1.25] (3) at (.8*3.2,-1){};
				\node[draw,circle,thick,scale=1.25] (4) at (.8*4.6,-1){};
								\node[draw,circle,thick,scale=1.25] (5) at (.8*3.2,1){};
																\node[draw,circle,thick,scale=1.25] (6) at (.8*4.6,1){};
			        \node at (.8*3.15,0)[draw,dashed, line width=2pt, ellipse, minimum width=120pt, minimum height=22pt,rotate=90,yshift=-1pt]{};
				\node at (.8*4.55,0)[draw,dashed, line width=2pt, ellipse, minimum width=120pt, minimum height=22pt,rotate=90,yshift=-1pt]{};
									\node   at (.8*3.15,-2.5) { 2};\node   at (.8*4.55,-2.5) { 1};

				\draw[thick] (0)--(1)--(2)--(3)--(4);
				\draw[thick] (2)--(5)--(6);
								\node (D0) at ($(0)+(-.2,.6)$) {$\mathbb{F}_3$};
				\node (D1) at ($(1)+(0,.6)$) {$\mathbb{F}_1$};
					\node (D2) at ($(2)+(0,.6)$) {$\mathbb{F}_1$};
						\node (D3) at ($(3)+(-.5,4.5)$) {\begin{tabular}{c} $D_3\to S$ is a  double cover of $ \mathbb{F}_3\to S$\\
						(but $D_3\to S'$ is an $\mathbb{F}_6$ )
						\end{tabular}
						};
							\node (D4) at ($(4)+(4,2)$)
							{\begin{tabular}{c} $D_4\to S$ is  a double cover of $ \mathbb{F}_4\to S$\\
						(but $D_3\to S'$ is an $\mathbb{F}_8$ )
						\end{tabular}
						};
				\draw[->,>=stealth',thick=1mm,dashed]  ($(D3)+(.5,-.4)$)--($(3)+(0,3.5)$);
				\draw[->,>=stealth',thick=1mm,dashed]  ($(D4)+(-.4,-.6)$)--($(4)+(.4,1)$);
								
	\end{tikzpicture}
\end{center}
		\caption{  
		Fibral divisors of a F$_4$-model as schemes over $S$ a curve of self-intersection $-5$. See Lemma \ref{Lem:frozenGeom}.
	The Stein factorization gives a morphism $f'_i:D_i\to S'$ with connected fibers and a finite morphism  $\pi: S'\to S$ that is a double cover branched at the two points $g_2=0$ of $S$. Hence $S'$ is also a rational curve.  
		The morphisms $f'_3:D_3\to S'$ and  $f'_4:D_4\to S'$ define, respectively, an  $\mathbb{F}_6$ and an $\mathbb{F}_8$ with base curve $S'$. 
\label{Fig:FibralDiv5}				}
		\end{figure}	

The case of an F$_4$-model with both representations  frozen has been recently studied in  \cite{DelZotto:2017pti}. 
 Such a model does not have any charged hypermultiplets as explained first  in \cite{Morrison:2012np}. 
 It follows from  the geometry of the crepant resolution, the generic fiber over the base curve $S$ does degenerate at two points $V(g_2)$. 
  Our description is consistent with the analyses of  \cite{GM1, GM2, Morrison:2012np} and the Hirzebruch surfaces identified in \cite{DelZotto:2017pti}.

We discuss  in detail the geometry of the fibral divisors of an F$_4$-model in the case where all representations are frozen. 
This means that $S$ is a curve of genus zero and self-intersection $-5$. For example, the base could be $\mathbb{F}_5$ or the total space of $\mathscr{O}_{\mathbb{P}^1}(-5)$. 
The key is Theorem  \ref{Thm:fibralGeom},which we specialize to this situation in the following lemma.

\begin{lem}\label{Lem:frozenGeom}
Let $S\subset B$ be a smooth rational curve of self-intersection $-5$. 
An F$_4$-model over $B$ with gauge group supported on $S$ has fibral divisors $D_0$, $D_1$, $D_2$, $D_3$ and $D_4$ such that 
$D_0\to S$, $D_1\to S$, $D_2\to S$ are respectively, Hirzebruch surfaces $\mathbb{F}_3$,  $\mathbb{F}_1$, and  $\mathbb{F}_1$. 
$D_3\to S$ and $D_4\to S$ are not Hirzebruch surfaces over $S$ but double covers of $\mathbb{F}_3$ and $\mathbb{F}_4$ Hirzebruch surfaces. 
Considering the  Stein factorization $D_3      \xrightarrow{f} S'\xrightarrow{\pi} S$ and $D_4      \xrightarrow{f} S'\xrightarrow{\pi} S$ where 
$S'\xrightarrow{\pi} S$ is a double cover of $S$ branched at two points. 
 The morphism 
$D_3      \xrightarrow{f} S'$ presents $D_3$ as an Hirzebruch surface $\mathbb{F}_6$ over $S'$ and $D_4      \xrightarrow{f} S'$ presents $D_4$ as a Hirzebruch surface $\mathbb{F}_8$ over $S'$.
\end{lem}
\begin{proof}
We use  Theorem \ref{Thm:fibralGeom}. 
The Calabi-Yau condition implies that $\mathscr{L}=\mathscr{O}(-K_B)$. 
D$_0$, D$_1$, and D$_2$ are Hirzebruch surfaces. We compute their degrees by intersection theory as follows. We recall that for a Hirzebruch surface $\mathbb{F}_{n}=\mathbb{F}_{-n}$. 
 We follow the following strategy. 
Given two line bundles $\mathscr{L}_1$ and $\mathscr{L}_2$ and a Hirzebruch surface $\mathbb{P}_S[\mathscr{L}_1\oplus\mathscr{L}_2]\to S$ over a smooth rational curve $S$, then the degree of the Hirzebruch surface is 
$\textnormal{deg} \mathscr{L}_1-\textnormal{deg} \mathscr{L}_2]=\int_S c_1(\mathscr{L}_1)-\int_S  c_1(\mathscr{L}_2)$.  
We know that D$_0$ is $\mathbb{P}_S[\mathscr{L}\oplus \mathscr{O}_S]$.
To compute the degree of this Hirzebruch surface, we compute $\int_S c_1(\mathscr{L})=-S\cdot K_B=-3$. that is, D$_0$ is an $\mathbb{F}_3$. 
We know that D$_1$ is $\mathbb{P}_S[\mathscr{L}^{\otimes 2} \oplus \mathscr{S}]$, which is 
isomorphic to $\mathbb{P}_S[\big(\mathscr{L}^{\otimes 2} \otimes \mathscr{S}^{-1}\big)\oplus \mathscr{O}_S]$
.

To compute the degree of this Hirzebruch surface, we compute $\int_S [2c_1(\mathscr{L})-S]=S\cdot (-2K_B-S)=-6+5=-1$. Hence $D_1\to S$ is an $\mathbb{F}_1$. 
D$_2$ is $\mathbb{P}_S[\mathscr{L}^{\otimes 3} \oplus \mathscr{S}^{\otimes 2}]$, which is a Hirzebruch surface of degree $S\cdot(-3K_B-2S)=-9+10=1$.  Hence, $D_2\to S$ is also an $\mathbb{F}_1$. 
As  is clear from the crepant resolution, the fibral divisors D$_3\to S$ and D$_4\to S$ are not Hirzebruch surfaces over $S$. 
Their fibers consist of two generically disconnected rational curves ramified over $R=V(s,g_2)$. 
They can be respectively described as double covers of  a  $\mathbb{P}_S[\mathscr{L}^{\otimes 4}\oplus \mathscr{S}^{\otimes 3}]$ and a  $\mathbb{P}_S[\mathscr{L}^{\otimes 2}\oplus \mathscr{S}^{\otimes 2}]$; These are isomorphic to, 
respectively,  $\mathbb{F}_3$ and $\mathbb{F}_4$. 
In each case, the branch locus of the double cover consists of two fibers of the Hirzebruch surface, and these are the fibers over $V(g_2)$. 

The absence of charged multiplets  is justified by the phenomena of frozen representations rather than the absence of degenerations supporting the weights of the representation $\mathbf{26}$. 
The Stein factorization of D$_3\to S$  is 
$$
D_3      \xrightarrow{f} S'\xrightarrow{\pi} S, 
$$ 
where $S'\xrightarrow{\pi} S$ is a double cover branched at two points and $f:D_3\to S'$ is a proper morphism with connected fibers. 
In particular, $f:D_3\to S'$ is a $\mathbb{P}^1$-bundle. Since $S'$ has genus zero, this is a Hirzebruch surface. 
The degree of this Hirzebruch surface is $\int_{D_3} S'^2$, which is $2\int_{\mathbb{F}_3} S^2=2\cdot (-3)=-6$. Hence $f:D_3\to S'$ is an $\mathbb{F}_6$-surface. 

We show in the same way that $f:D_3\to S'$ is an $\mathbb{F}_8$-surface.  
In other words, while  $D_3\to S$ and $D_4\to S$ are not $\mathbb{P}^1$-bundles over $S$, $D_3\to S'$ and $D_3\to S'$ are $\mathbb{P}^1$-bundles over $S'$. 
As discussed in the appendix, with respect to $S'$, $D_3$ and D$_4$ have the structure of an $\mathbb{F}_6$ and an $\mathbb{F}_8$ since the base curve $S'$ has self-intersection

\end{proof}

 \section{Conclusion} \label{sec:conclusion}

In this paper, we studied the geometry of F$_4$-models. Our starting point is a  singular Weierstrass model characterized by the valuations with respect to a smooth divisor $S=V(s)$ given by  Step 8 of Tate's algorithm:
$$
v_S(a_1)\geq 1, \quad v_S(a_2)\geq 2, \quad v_S(a_3)\geq 2, \quad v_S(a_4)\geq 3, \quad v_S(a_6)\geq 3,\quad \quad 
v_S(b_6)=4.$$ 
The last condition ensures that the   polynomial $Q(T)=T^2 + a_{3,2}  T-a_{6,4}$ has two distinct solutions modulo $s$. 
 We focus on the case where $Q(T)$ has no rational solutions modulo $s$. The generic fiber over $S$ is called a fiber of type IV$^{*\text{ns}}$. 
Without loss of generality, the Weierstrass model can be written in the following canonical form: 
$$
y^2z= x^3 + a_{4,3+\alpha}  s^{3+\alpha} x z^2+ a_{6,4} s^4 z^3,\quad \alpha\in\mathbb{Z}_{\geq 0}.
$$
A crepant resolution of this Weierstrass model is called an F$_4$-model. Such elliptic fibrations are used to engineer F$_4$ gauge theories in F-theory and M-theory. 
While it is common to take $\alpha=0$  in the F-theory literature, here,  we keep $\alpha$ unfixed to keep the geometry as general as possible.
This allows us to cover local enhancements of the type F$_4$ $\to$ E$_7$ and F$_4$ $\to$ E$_8$ over $s=a_{6,4}=0$, depending on the valuation of $a_4$. 
While the generic fiber is a twisted affine Dynkin diagram $\tilde{\text{F}}_4^t$, the geometric fiber is the affine $\tilde{E}_6$ diagram.
Thus we have the natural enhancement F$_4$ $\to$ E$_6$,  which is non-local and appears over any closed point of $S$ away from $a_{6,4}=0$.

The crepant resolution that we have considered consists of a sequence of  four blowups centered at regular monomial ideals. We answered several questions regarding the geometry and topology of the resulting smooth  elliptic fibration. 
 In particular, we identified the geometry of the fibral divisors of the F$_4$-model  as $\mathbb{P}^1$-bundles for the three divisors D$_0$, D$_1$, and D$_2$, while the remainders, namely D$_3$ and D$_4$, are  double-coverings of $\mathbb{P}^1$-bundles with discriminant locus $a_{6,4}=0$. This is illustrated in Figure \ref{Fig:FibralDiv}. The singular fibers of these conic bundles consist of double lines and play an important role in determining the geometry of the singular fiber over $s=a_{6,4}=0$, which is a non-Kodaira fiber of type $1-2-3-4-2$. 
This fiber can be  thought of as an incomplete E$_7$ or E$_8$ if  $v(a_4)=3$ or $v(a_4)\geq 4$, respectively. 
 
To identify the representation associated with this singular fiber over $s=a_{6,4}=0$, we computed the intersection numbers of the new rational curves with the fibral divisors. 
These intersection numbers are interpreted as weights of F$_4$. We identified the corresponding representation as  the $\mathbf{26}$ of F$_4$. 
 We also compute the triple intersection number of the fibral divisors. 
 We finally specialize to the case of  Calabi-Yau threefolds. The Euler characteristic of an F$_4$-model and the Hodge numbers in the Calabi-Yau threefold case have been presented  in  \cite{MMP1}.
We also computed  the linear form induced by the second Chern class. 

In the final section, we studied details of M-theory compactified on a Calabi-Yau threefold that is an F$_4$-model, for which 
the resulting theory is a five-dimensional gauge theory with eight supercharges. 
Such a theory has vector multiplets characterized by a cubic prepotential. The classical part of the cubic prepotential vanishes but there is a quantum correction coming from an exact one-loop contribution. 
This one-loop term depends explicitly on the number of charged hypermultiplets. 
 It is known that this correction term matches exactly with the triple intersection numbers of the fibral divisors of the elliptic Calabi-Yau threefold. 
 We computed the number of  hypermultiplets in the adjoint representation and in the fundamental representation via a direct comparison:
$$
n_{\mathbf{52}}=g, \quad n_{\mathbf{26}}=5 (1-g) + S^2.
$$
 We checked that they satisfy the genus formula of Aspinwall-Katz-Morrison-- here they are derived from the triple intersection numbers. The same number were computed by Grassi and Morrison using Witten's genus formula. 
  With the  knowledge of the Hodge numbers, matter representation, and their multiplicities, we checked explicitly that these data are compatible with a six-dimensional $(1,0)$  supergravity theory free of gravitational, gauge, and mixed anomalies.

We also computed the weights of vertical curves of an F$_4$ models and proved that over the locus $V(s,a_{6,4})$ the generic fiber over the divisor $S$ has weights of the quasi-minuscule representation $\mathbf{26}$. 
We introduced the notion of a {\em frozen representation}, which explains that the existence of geometric weights carried by vertical curves does not imply the existence of hypermultiplets charged under the corresponding representation. 
If the base is a surface, 
the divisor $S$ is a curve. The  representation $\mathbf{26}$ is frozen if and only if  $S$ has genus zero and self-intersection $-5$.

\section*{Acknowledgements}
The authors are grateful to  Paolo Aluffi, Michele del Zotto, Antonella Grassi, Jim Halverson, Jonathan Heckman, Ravi Jagadeesan, Sheldon Katz, Craig Lawrie, David Morrison,  Kenji Matsuki, Shu-Heng Shao, and Shing-Tung Yau  for helpful discussions. 
The authors are thankful to all the participants of Series of Lectures by Kenji Matsuki in March 2017 on   ``Beginners Introduction to the Minimal Model Program'' hold at Northeastern University supported by Northeastern University and  the National Science Foundation (NSF) grant DMS-1603247. 
M.E. is supported in part by the National Science Foundation (NSF) grant DMS-1701635  ``Elliptic Fibrations and String Theory''.
P.J.  is  supported by NSF grant PHY-1067976. 
 P.J.  is thankful to Cumrun Vafa  for his guidance and constant support. 
M.J.K. would like to acknowledge partial support from NSF grant PHY-1352084.
M.J.K. would like to extend her gratitude to  Daniel Jafferis for his tutelage and continued support.

\begin{table}[htb]
\begin{center}
$			
\begin{array}{|c|c| c  |} \hline
\vrule width 0pt height 3ex 
 \text{Fiber Type} & \text{ Dual graph  } & \text{Dual graph of Geometric fiber } \\\hline
 			  \begin{array}{c}
\text{I}^{\text{ns}}_{3}, \text{IV}^{\text{ns}}\\
						  \\
						 \widetilde{\text{A}}_{1}  \\
						 \\
						 \end{array}

 &			\scalebox{1}{$\begin{array}{c}\begin{tikzpicture}
				\node[draw,circle,thick,scale=1.25,label=above:{1}, fill=black] (1) at (0,0){};
				\node[draw,circle,thick,scale=1.25,label=above:{1}] (2) at (1.3,0){};
				\draw[thick] (0.15,0.1) --++ (.95,0);
				\draw[thick] (0.15,-0.09) --++ (.95,0);
				
			\end{tikzpicture}\end{array}
		$} &
		\scalebox{1}{$\begin{array}{c}\begin{tikzpicture}
				\node[draw,circle,thick,scale=1.25,label=above:{1}, fill=black] (1) at (-.6,0){};
				\node[draw,circle,thick,scale=1.25,label=above:{1}] (2) at (45:.6){};
				\node[draw,circle,thick,scale=1.25,label=below:{1}] (3) at (-45:.6){};
				\draw[thick] (0,0) --(1);
				\draw[thick] (0,0) --(2);
				\draw[thick] (0,0) --(3);
								\draw[<->,>=stealth',semithick,dashed]  (.8,-.5) arc (-30:30:1cm);
			\end{tikzpicture}\end{array}$}

				\\\hline

						 \begin{array}{c}

						  \text{I}^{*\text{ns}}_{\ell-3}\\
						  \\
						 \widetilde{\text{B}}_{\ell}^t  \\
						 \\
						 (\ell\geq 3)
						 \end{array}
						 &\scalebox{1}{$\begin{array}{c} \begin{tikzpicture}
				\node[draw,circle,thick,scale=1.25,fill=black,label=above:{1}] (1) at (-.1,.7){};
				\node[draw,circle,thick,scale=1.25,label=above:{1}] (2) at (-.1,-.7){};	
				\node[draw,circle,thick,scale=1.25,label=above:{2}] (3) at (1,0){};
				\node[draw,circle,thick,scale=1.25,label=above:{2}] (4) at (2.1,0){};
				\node[draw,circle,thick,scale=1.25,label=above:{2}] (5) at (3.3,0){};	
				\node[draw,circle,thick,scale=1.25,label=above:{1}] (6) at (4.6,0){};	
				\draw[thick] (1) to (3) to (4);
				\draw[thick] (2) to (3);
				\draw[ultra thick, loosely dotted] (4) to (5) {};
				\draw[thick] (3.5,-0.05) --++ (.9,0){};
				\draw[thick] (3.5,+0.05) --++ (.9,0){};
				\draw[thick]
					(3.9,0) --++ (60:.25)
					(3.9,0) --++ (-60:.25);
			\end{tikzpicture}\end{array}$}&
			
\scalebox{1}{$\begin{array}{c} \begin{tikzpicture}
				\node[draw,circle,thick,scale=1.25,fill=black,label=above:{1}] (1) at (-.1,.7){};
				\node[draw,circle,thick,scale=1.25,label=above:{1}] (2) at (-.1,-.7){};	
				\node[draw,circle,thick,scale=1.25,label=above:{2}] (3) at (1,0){};
				\node[draw,circle,thick,scale=1.25,label=above:{2}] (4) at (2.1,0){};
				\node[draw,circle,thick,scale=1.25,label=above:{2}] (5) at (3.3,0){};	
				\node[draw,circle,thick,scale=1.25,label=above:{1}] (6) at (4.6,.7){};	
								\node[draw,circle,thick,scale=1.25,label=above:{1}] (7) at (4.6,-.7){};	
				\draw[thick] (1) to (3) to (4);
				\draw[thick] (2) to (3);
				\draw[ultra thick, loosely dotted] (4) to (5) {};
				\draw[thick] (5) to (6); 
				\draw[thick] (5) to (7);
				\draw[<->,>=stealth',semithick,dashed]  (5,-.5) arc (-30:30:1.2cm);
			\end{tikzpicture}\end{array}$}
			
			\\\hline
			\begin{array}{c}
			     \text{I}_{2\ell+2}^{\text{ns}}
			\\
			\\
			\vrule width 0pt height 3ex 
 \widetilde{\text{C}}_{\ell+1}^t  \\
						 \\
						 (\ell\geq 1)
 \end{array}
   &
			\scalebox{.95}{$\begin{array}{c} \begin{tikzpicture}
				\node[draw,circle,thick,scale=1.25,fill=black,label=above:{1}] (1) at (-.2,0){};
				\node[draw,circle,thick,scale=1.25,label=above:{1}] (3) at (.8,0){};
				\node[draw,circle,thick,scale=1.25,label=above:{1}] (4) at (1.8,0){};
				\node[draw,circle,thick,scale=1.25,label=above:{1}] (5) at (2.8,0){};	
				\node[draw,circle,thick,scale=1.25,label=above:{1}] (6) at (3.8,0){};	
				\node[draw,circle,thick,scale=1.25,label=above:{1}] (7) at (4.8,0){};	
				\draw[thick]   (3) to (4);
				\draw[thick] (5) to (6);
				\draw[ultra thick, loosely dotted] (4) to (5) {};
				\draw[thick] (4.,-0.05) --++ (.6,0){};
				\draw[thick] (4.,+0.05) --++ (.6,0){};
								\draw[thick] (-.2,-0.05) --++ (.8,0){};
				\draw[thick] (-.2,+0.05) --++ (.8,0){};

				\draw[thick]
					(4.4,0) --++ (-120:.25)
					(4.4,0) --++ (120:.25);
					\draw[thick]
					(0.2,0) --++ (-60:.25)
					(0.2,0) --++ (60:.25);
			\end{tikzpicture}\end{array}$}

			&

			\scalebox{.95}{$\begin{array}{c} \begin{tikzpicture}
				\node[draw,circle,thick,scale=1.25,fill=black,label=above:{1}] (1) at (-.2,0){};	
				\node[draw,circle,thick,scale=1.25,label=below:{1}] (3a) at (.8,-.8){};
				\node[draw,circle,thick,scale=1.25,label=below:{1}] (4a) at (1.8,-.8){};
				\node[draw,circle,thick,scale=1.25,label=below:{1}] (5a) at (2.8,-.8){};	
				\node[draw,circle,thick,scale=1.25,label=below:{1}] (6a) at (3.8,-.8){};	
								\node[draw,circle,thick,scale=1.25,label=above:{1}] (3b) at (.8,.8){};
				\node[draw,circle,thick,scale=1.25,label=above:{1}] (4b) at (1.8,.8){};
				\node[draw,circle,thick,scale=1.25,label=above:{1}] (5b) at (2.8,.8){};	
				\node[draw,circle,thick,scale=1.25,label=above:{1}] (6b) at (3.8,.8){};	

				\node[draw,circle,thick,scale=1.25,label=above:{1}] (7) at (4.8,0){};	
				\draw[thick]   (4b)--(3b)--(1)--(3a)--(4a);
				\draw[thick]   (5b)--(6b)--(7)--(6a)--(5a);
				\draw[ultra thick, loosely dotted] (4a) to (5a) {};
								\draw[ultra thick, loosely dotted] (4b) to (5b) {};
\draw[<->,>=stealth',semithick,dashed] ($(4a)+(0,0.3)$) --($(4b)-(0,0.3)$) {};
\draw[<->,>=stealth',semithick,dashed] ($(5a)+(0,0.3)$) --($(5b)-(0,0.3)$) {};
\draw[<->,>=stealth',semithick,dashed] ($(6a)+(0,0.3)$) --($(6b)-(0,0.3)$) {};
\draw[<->,>=stealth',semithick,dashed] ($(3a)+(0,0.3)$) --($(3b)-(0,0.3)$) {};
			\end{tikzpicture}\end{array}$}

			\\\hline

			\begin{array}{c}
			     \text{I}_{2\ell+3}^{\text{ns}}
			\\
			\\
			\vrule width 0pt height 3ex 
 \widetilde{\text{C}}_{\ell+1}^t   \\
						 \\
						 (\ell\geq 1)
 \end{array}
   &
			\scalebox{.95}{$\begin{array}{c} \begin{tikzpicture}
				\node[draw,circle,thick,scale=1.25,fill=black,label=above:{1}] (1) at (-.2,0){};
				\node[draw,circle,thick,scale=1.25,label=above:{1}] (3) at (.8,0){};
				\node[draw,circle,thick,scale=1.25,label=above:{1}] (4) at (1.8,0){};
				\node[draw,circle,thick,scale=1.25,label=above:{1}] (5) at (2.8,0){};	
				\node[draw,circle,thick,scale=1.25,label=above:{1}] (6) at (3.8,0){};	
				\node[draw,circle,thick,scale=1.25,label=above:{1}] (7) at (4.8,0){};	
				\draw[thick]   (3) to (4);
				\draw[thick] (5) to (6);
				\draw[ultra thick, loosely dotted] (4) to (5) {};
				\draw[thick] (4.,-0.05) --++ (.6,0){};
				\draw[thick] (4.,+0.05) --++ (.6,0){};
								\draw[thick] (-.2,-0.05) --++ (.8,0){};
				\draw[thick] (-.2,+0.05) --++ (.8,0){};

				\draw[thick]
					(4.4,0) --++ (-120:.25)
					(4.4,0) --++ (120:.25);
					\draw[thick]
					(0.2,0) --++ (-60:.25)
					(0.2,0) --++ (60:.25);
			\end{tikzpicture}\end{array}$}

			&
			
				\scalebox{.95}{$\begin{array}{c} \begin{tikzpicture}
				\node[draw,circle,thick,scale=1.25,fill=black,label=above:{1}] (1) at (-.2,0){};	
				\node[draw,circle,thick,scale=1.25,label=below:{1}] (3a) at (.8,-.8){};
				\node[draw,circle,thick,scale=1.25,label=below:{1}] (4a) at (1.8,-.8){};
				\node[draw,circle,thick,scale=1.25,label=below:{1}] (5a) at (2.8,-.8){};	
				\node[draw,circle,thick,scale=1.25,label=below:{1}] (6a) at (3.8,-.8){};	
								\node[draw,circle,thick,scale=1.25,label=above:{1}] (3b) at (.8,.8){};
				\node[draw,circle,thick,scale=1.25,label=above:{1}] (4b) at (1.8,.8){};
				\node[draw,circle,thick,scale=1.25,label=above:{1}] (5b) at (2.8,.8){};	
				\node[draw,circle,thick,scale=1.25,label=above:{1}] (6b) at (3.8,.8){};	

				\node[draw,circle,thick,scale=1.25,label=below:{1}] (7a) at (4.8,-.8){};
				\node[draw,circle,thick,scale=1.25,label=above:{1}] (7b) at (4.8,.8){};	
				\draw[thick]   (4b)--(3b)--(1)--(3a)--(4a);
				\draw[thick]   (5b)--(6b)--(7b)--(7a)--(6a)--(5a);
				\draw[ultra thick, loosely dotted] (4a) to (5a) {};
								\draw[ultra thick, loosely dotted] (4b) to (5b) {};

\draw[<->,>=stealth',semithick,dashed] ($(4a)+(0,0.3)$) --($(4b)-(0,0.3)$) {};
\draw[<->,>=stealth',semithick,dashed] ($(5a)+(0,0.3)$) --($(5b)-(0,0.3)$) {};
\draw[<->,>=stealth',semithick,dashed] ($(6a)+(0,0.3)$) --($(6b)-(0,0.3)$) {};
\draw[<->,>=stealth',semithick,dashed] ($(3a)+(0,0.3)$) --($(3b)-(0,0.3)$) {};
\draw[<->,>=stealth',semithick,dashed] (5.2,-1) arc (-70:70:1) {};
			\end{tikzpicture}\end{array}$}

			\\\hline
			\begin{array}{c}
			\text{IV}^{* \text{ns}}
			\\ 
			\vrule width 0pt height 3ex 
\widetilde{\text{F}}_4^t      
\end{array}
&			

\scalebox{1}{$\begin{array}{c}\begin{tikzpicture}
				\node[draw,circle,thick,scale=1.25,fill=black,label=above:{1}] (1) at (0,0){};
				\node[draw,circle,thick,scale=1.25,label=above:{2}] (2) at (1,0){};
				\node[draw,circle,thick,scale=1.25,label=above:{3}] (3) at (2,0){};
				\node[draw,circle,thick,scale=1.25,label=above:{2}] (4) at (3,0){};
				\node[draw,circle,thick,scale=1.25,label=above:{1}] (5) at (4,0){};
				\draw[thick] (1) to (2) to (3);
				\draw[thick]  (4) to (5);
				\draw[thick] (2.2,0.05) --++ (.6,0);
				\draw[thick] (2.2,-0.05) --++ (.6,0);
				\draw[thick]
					(2.4,0) --++ (60:.25)
					(2.4,0) --++ (-60:.25);
			\end{tikzpicture}\end{array}$} &
			\scalebox{1}{$\begin{array}{c}\begin{tikzpicture}
				\node[draw,circle,thick,scale=1.25,fill=black,label=below:{1}] (0) at (0,0){};
				\node[draw,circle,thick,scale=1.25,label=below:{2}] (1) at (.8,0){};
				\node[draw,circle,thick,scale=1.25,label=below:{3}] (2) at (.8*2,0){};
				\node[draw,circle,thick,scale=1.25,label=below:{2}] (3) at (.8*3,0){};
				\node[draw,circle,thick,scale=1.25,label=below:{1}] (4) at (.8*4,0){};
								\node[draw,circle,thick,scale=1.25,label=left:{2}] (5) at (.8*2,.8*1){};
																\node[draw,circle,thick,scale=1.25,label=above:{1}] (6) at (.8*2,.8*2){};
				\draw[thick] (0)--(1)--(2)--(3)--(4);
				\draw[thick] (2)--(5)--(6);
				\draw[<->,>=stealth',semithick,dashed]  (2.5,0.3) arc (25:65:1.3cm);
				\draw[<->,>=stealth',semithick,dashed]  (3.3,0.3) arc (25:65:3cm);

			\end{tikzpicture}\end{array}$}
		
			\\\hline
			\begin{array}{c}
			\text{I}^{*\text{ss}}_{0}
			\\
			\widetilde{\text{B}}_3^t
\end{array}
 &			\scalebox{1}{$\begin{array}{c}\begin{tikzpicture}
				\node[draw,circle,thick,scale=1.25,label=below:{1}] (1) at (0,-.5){};
				\node[draw,circle,thick,scale=1.25,label=above:{2}] (2) at (1.3,0){};
				\node[draw,circle,thick,scale=1.25,label=above:{1}] (3) at (2.6,0){};
								\node[draw,circle,thick,scale=1.25,label=above:{1}, fill=black] (4) at (0,.5){};
				\draw[thick] (1) to (2);\draw[thick] (2) to (4);
				\draw[thick] (1.5,0.09) --++ (.9,0);
				\draw[thick] (1.5,-0.09) --++ (.9,0);
				\draw[thick]
					(1.9,0) --++ (60:.25)
					(1.9,0) --++ (-60:.25);
			\end{tikzpicture}\end{array}
		$} & 
		\scalebox{.8}{$\begin{array}{c}\begin{tikzpicture}
				\node[draw,circle,thick,scale=1.25,label=30:{2}] (0) at (0,0){};
				\node[draw,circle,thick,scale=1.25,label=below:{1}, fill=black] (1) at (-1,0){};
				\node[draw,circle,thick,scale=1.25,label=below:{1}] (2) at (1,0){};
				\node[draw,circle,thick,scale=1.25,label=above:{1}] (3) at (90:1){};      								\node[draw,circle,thick,scale=1.25,label=below:{1}] (4) at (90:-1){};
				\draw[thick] (0) to (1);
				\draw[thick] (0) to (2);
				\draw[thick] (0) to (3);
				\draw[thick] (0) to (4);
								\draw[<->,>=stealth',semithick,dashed]  (1.1,0.3) arc (25:65:1.7cm);
			\end{tikzpicture}\end{array}
		$}

				\\\hline

			\begin{array}{c}
			\text{I}^{*\text{ns}}_{0}
			\\
			\widetilde{\text{G}}_2^t
\end{array}
 &			\scalebox{1}{$\begin{array}{c}\begin{tikzpicture}
				\node[draw,circle,thick,scale=1.25,label=above:{1}, fill=black] (1) at (0,0){};
				\node[draw,circle,thick,scale=1.25,label=above:{2}] (2) at (1.3,0){};
				\node[draw,circle,thick,scale=1.25,label=above:{1}] (3) at (2.6,0){};
				\draw[thick] (1) to (2);
				\draw[thick] (1.5,0.09) --++ (.9,0);
				\draw[thick] (1.5,-0.09) --++ (.9,0);
				\draw[thick] (1.5,0) --++ (.9,0);
				\draw[thick]
					(1.9,0) --++ (60:.25)
					(1.9,0) --++ (-60:.25);
			\end{tikzpicture}\end{array}
		$} & 
		\scalebox{.8}{$\begin{array}{c}\begin{tikzpicture}
				\node[draw,circle,thick,scale=1.25,label=30:{2}] (0) at (0,0){};
				\node[draw,circle,thick,scale=1.25,label=above:{1}, fill=black] (1) at (-1,0){};
				\node[draw,circle,thick,scale=1.25,label=right:{1}] (2) at (1,0){};
				\node[draw,circle,thick,scale=1.25,label=above:{1}] (3) at (90:1){};
								\node[draw,circle,thick,scale=1.25,label=below:{1}] (4) at (90:-1){};
				\draw[thick] (0) to (1);
				\draw[thick] (0) to (2);
				\draw[thick] (0) to (3);
				\draw[thick] (0) to (4);
				\draw[<->,>=stealth',semithick,dashed]  (1.1,0.3) arc (25:65:1.7cm);
				\draw[<->,>=stealth',semithick,dashed]  (1.1,-0.3) arc (-25:-65:1.7cm);
			\end{tikzpicture}\end{array}
		$}

				\\\hline

				\end{array}$
	\end{center}
	\caption{
{ Dual graphs for singular fibers elliptic fibrations with non-geometrically irreducible fiber components.}  \label{Table:DualGraph}
}
\end{table}

\clearpage

\appendix 

\section{Weierstrass models and Deligne's formulaire }
\label{sec:Wmodel}

In this section, we follow the notation of \cite{Formulaire}. 
Let  $\mathscr{L}$ be a line bundle over  a normal quasi-projective variety  $B$.  We define the projective bundle (of lines)
\begin{equation}
\pi: X_0=\mathbb{P}_B[\mathscr{O}_B\oplus \mathscr{L}^{\otimes 2}\oplus \mathscr{L}^{\otimes 3}]\longrightarrow B.
\end{equation} 
The relative projective coordinates of $X_0$ over $B$ are denoted $[z:x:y]$,  where $z$, $x$, and $y$ are defined  by the natural injection of 
 $\mathscr{O}_B$,   $\mathscr{L}^{\otimes 2}$, and $\mathscr{L}^{\otimes 3}$ into $\mathscr{O}_B\oplus \mathscr{L}^{\otimes 2}\oplus \mathscr{L}^{\otimes 3}$, respectively. Hence, 
  $z$ is a section of $\mathscr{O}_{X_0}(1)$, $x$ is a section of $\mathscr{O}_{X_0}(1)\otimes \pi^\ast \mathscr{L}^{\otimes 2}$, and
$y$ is a section of  $\mathscr{O}_{X_0}(1)\otimes \pi^\ast \mathscr{L}^{\otimes 3}$.

\begin{defn}
 A  Weierstrass model is an elliptic fibration $\varphi: Y\to B$  cut out by the zero locus of  a section of the  
line bundle $\mathscr{O}(3)\otimes \pi^\ast \mathscr{L}^{\otimes 6}$ in $X_0$. 
\end{defn}
The most general Weierstrass equation is written in the notation of Tate as
\begin{equation}
y^2z+ a_1 xy z + a_3  yz^2 -(x^3+ a_2 x^2 z + a_4 x z^2 + a_6 z^3) =0,
\end{equation} 
where $a_i$ is a section of $\pi^\ast \mathscr{L}^{\otimes i}$. 
The line bundle $\mathscr{L}$ is called the {\em fundamental line bundle} of the Weierstrass model $\varphi:Y\to B$. It can be defined directly from $Y$ as 
$\mathscr{L}=R^1 \varphi_\ast Y$. 
Following Tate and Deligne, we introduce the following quantities 
\begin{align}
\begin{cases}
b_2 &= a_1^2 + 4 a_2\\
b_4 &= a_1 a_3 + 2 a_4\\
b_6 &= a_3^2 + 4 a_6\\
b_8 &= a_1^2 a_6 - a_1 a_3 a_4 + 4 a_2 a_6 + a_2 a_3^2 - a_4^2\\
c_4 &= b_2^2 - 24 b_4\\
c_6 &= -b_2^3 + 36 b_2 b_4 - 216 b_6\\
\Delta &= -b_2^2 b_8 - 8 b_4^3 - 27 b_6^2 + 9 b_2 b_4 b_6\\
j& = {c_4^3}/{\Delta}
\end{cases}
\end{align}
The  $b_i$ ($i=2,3,4,6)$ and $c_i$   ($i=4,6$) are  sections of $\pi^\ast \mathscr{L}^{\otimes i}$. 
The discriminant $\Delta$ is a section of $\pi^\ast \mathscr{L}^{\otimes 12}$. 
They satisfy the two relations
\begin{align}
1728 \Delta=c_4^3-c_6^2, \quad 4b_8 = b_2 b_6 - b_4^2.
\end{align}
Completing the square in $y$ gives 
\begin{equation}
zy^2 =x^3 +\tfrac{1}{4}b_2 x^2 + \tfrac{1}{2} b_4 x + \tfrac{1}{4} b_6.
\end{equation}
Completing the cube in $x$ gives the short form of the Weierstrass equation
\begin{equation}
zy^2 =x^3 -\tfrac{1}{48} c_4 x z^2 -\tfrac{1}{864} c_6 z^3.
\end{equation}

\section{Representation theory of F$_4$}

F$_4$ is studied in Planche VIII of \cite{Boubaki.Lie46}.
The exceptional group F$_4$ has rank $4$, Coxeter number $12$, dimension $52$, and a trivial center. 
Its root system consists of $48$ roots, half of which are short roots. The Weyl group $W(\text{F}_4)$ of F$_4$ 
is the semi-direct product of the symmetric group $\mathfrak{S}_3$
with the semi-direct product of the symmetric group $\mathfrak{S}_4$ and $(\mathbb{Z}/2\mathbb{Z})^3$. Hence, $W(F_4)$
has dimension $3! 4! 2^3=2^7\times 3^2=1152$.  $W(\text{F}_4)$ is also a solvable group isomorphic to the symmetry group of the $24$-cell. 
The long roots of F$_4$ form a sublattice of index $4$.  
The outer automorphism group of F$_4$ is trivial. Hence, F$_4$ has neither complex nor quaternionic representations, and all its representations are real. 
Its smallest representation has dimension $26$ and is usually called the fundamental representation of F$_4$.  
F$_4$ is not simply laced and can be described from E$_6$ by a $\mathbb{Z}/2\mathbb{Z}$ folding.

\begin{figure}[htb]
\begin{center}
\begin{tikzpicture}
				\node[draw,circle,thick,scale=1.25,label=above:{$\alpha_1$}] (2) at (1,0){};
				\node[draw,circle,thick,scale=1.25, label=above:{$\alpha_2$}] (3) at (2,0){};
				\node[draw,circle,thick,scale=1.25,label=above:{$\alpha_3$}] (4) at (3,0){};
				\node[draw,circle,thick,scale=1.25, label=above:{$\alpha_4$}] (5) at (4,0){};
				\draw[thick] (2) to (3);
				\draw[thick]  (4) to (5);
				\draw[thick] (2.2,0.05) --++ (.6,0);
				\draw[thick] (2.2,-0.05) --++ (.6,0);
				\draw[thick]
					(2.6,0) --++ (-120:.25)
					(2.6,0) --++ (120:.25);
	\end{tikzpicture}
	\quad \quad \quad 
	$
\begin{pmatrix}
2 & -1 & 0 & 0 \\
-1& 2 & -2  & 0 \\
0 & -1 & 2 & -1 \\
0 & 0 & -1 & 2
\end{pmatrix}
$
	\end{center}
	\caption{Dynkin diagram and Cartan matrix of of F$_4$}
		\end{figure}
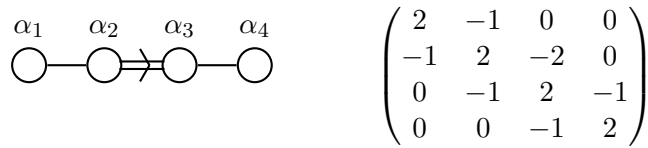

\begin{center}
\begin{tabular}{|c|c|}
\hline
Highest weight & Dimension \\
\hline
$(1,0,0,0)$ & $52$ \\
$(0,1,0,0)$ & $1274$\\
$(0,0,1,0)$ & $273$\\
$(0,0,0,1)$ & $26$ \\
\hline 
\end{tabular}
\end{center}

F$_4$ contains both B$_3$ and C$_3$, as  is clear by removing the first or last node.  
It is less trivial to see that F$_4$ also contains  B$_4$. One way to see it is to remember that F$_4=su(3,\mathbb{O})$ while $\mathfrak{so}_9=\mathfrak{su}(2,\mathbb{O})$.
The coset manifold F$_4/\text{Spin}(9)$ is the octonionic projective plane $\mathbb{O}\mathbb{P}^2$.
It follows that we have the isomorphism \cite{Bernadoni,Figueroa}
$$
\pi_i(\text{F}_4)=\pi_i(\text{Spin}(9)), \quad i\leq 6.
$$

The compact real form of F$_4$ can be described as the automorphism group of the Jordan Lie algebra $J_3$ of dimension $27$ \cite{Bernadoni}. 
Maybe a more geometrically familiar picture is to describe the  compact real form of F$_4$ as the {\em Killing superalgebra} of the $8$-sphere $S^{8}$
\cite{Figueroa}. 
In this form, the Lie algebra F$_4$ decomposes (as a vector space) as the direct sum of the Lie algebra B$_4\cong \mathfrak{so}_9$ and its  spin representation.

\section{A double cover of a ruled surface branch along $2b$ fibers}

A ruled surface is by definition a $\mathbb{P}^1$-bundle over a smooth curve of genus $g$.   
Let $p: Y\to S$ be the projection of a ruled  surface $Y$ to its base curve $S$. There exists a non-negative  number $n$ such that $S^2=-n$.
We will now construct the double cover $X$ of $Y$ branched along $2b$ fibers. 

We denote the projection map as $\pi: X\to Y$. 
We assume that the ramification locus of $\pi$ consists of $2b$ distinct fibers of $Y$. Then  $X$ is smooth and defines a flat double cover of $Y$. 
  Let $f$ be the class of a  generic fiber of $Y$ and define $\ell$ such that $2\ell=\pi^* f$.  We have
\begin{align}
K_X &=\pi^* K_Y+ b \pi^* f=-2 \overline{S}+(2(2g+b-1)-2-2n) \ell \nonumber \\
K_X^2&= 8(2-2g-b), \quad c_2 =4(2-2g-b), \quad
\chi(\mathscr{O}_X)=2 -2g-b. \nonumber
\end{align}
The curve $\overline{S}=\pi^* S$  is a double cover of $S$ branched at $2b$ distinct points. Hence,  $\overline{S}$ is a smooth curve of genus $2g+b-1$. 
The self-intersection of $\overline{S}$ in $X$ is $-2n$ by a pushforward argument: 
  $$
\int_X  \overline{S}^2=\int_Y \pi_* \overline{S}^2=2\int_Y S^2=-2n,
  $$
  where we used the fact that $\pi$ is a finite map of degree $2$. 
The $\mathbb{P}^1$-bundle over $S$ pulls back to a $\mathbb{P}^1$-bundle over $\overline{S}$. Hence, 
$X$ is a ruled surface over a curve $\overline{S}$ of genus $(2g+b-1)$  and self-intersection $-2n$.  
By the universal property of the Stein factorization, $X\overset{f'}{\longrightarrow } \overline{S}\overset{\pi}{\longrightarrow } S$ is the Stein factorization of 
$X\to S$.


\begin{thebibliography}{10}
%



\bibitem{AE1}
P.~Aluffi and M.~Esole.
\newblock {Chern class identities from tadpole matching in type IIB and
  F-theory}.
\newblock {\em JHEP}, 03:032, 2009.

\bibitem{AE2}
P.~Aluffi and M.~Esole.
\newblock {New Orientifold Weak Coupling Limits in F-theory}.
\newblock {\em JHEP}, 02:020, 2010.

\bibitem{Anderson:2016ler} 
  L.~B.~Anderson, X.~Gao, J.~Gray and S.~J.~Lee,
  ``Tools for CICYs in F-theory,''
  JHEP {\bf 1611}, 004 (2016)
  doi:10.1007/JHEP11(2016)004
  [arXiv:1608.07554 [hep-th]].


\bibitem{Arras:2016evy} 
  P.~Arras, A.~Grassi and T.~Weigand,
  ``Terminal Singularities, Milnor Numbers, and Matter in F-theory,''
  arXiv:1612.05646 [hep-th].


\bibitem{Aspinwall:1996nk} 
  P.~S.~Aspinwall and M.~Gross,
  ``The SO(32) heterotic string on a K3 surface,''
  Phys.\ Lett.\ B {\bf 387}, 735 (1996)
  doi:10.1016/0370-2693(96)01095-7
  [hep-th/9605131].
\bibitem{Aspinwall:2000kf} 
  P.~S.~Aspinwall, S.~H.~Katz and D.~R.~Morrison,
  ``Lie groups, Calabi-Yau threefolds, and F theory,''
  Adv.\ Theor.\ Math.\ Phys.\  {\bf 4}, 95 (2000)
  [hep-th/0002012].
  
  \bibitem{Atiyah:2001qf} 
  M.~Atiyah and E.~Witten,
  ``M theory dynamics on a manifold of G(2) holonomy,''
  Adv.\ Theor.\ Math.\ Phys.\  {\bf 6}, 1 (2003)
  [hep-th/0107177].
  
  \bibitem{Avramis:2005hc} 
  S.~D.~Avramis and A.~Kehagias,
  ``A Systematic search for anomaly-free supergravities in six dimensions,''
  JHEP {\bf 0510}, 052 (2005)
  doi:10.1088/1126-6708/2005/10/052
  [hep-th/0508172].

\bibitem{Batyrev.Betti}
V.~V. Batyrev.
\newblock Birational {C}alabi-{Y}au {$n$}-folds have equal {B}etti numbers.
\newblock In {\em New trends in algebraic geometry ({W}arwick, 1996)}, volume
  264 of {\em London Math. Soc. Lecture Note Ser.}, pages 1--11. Cambridge
  Univ. Press, Cambridge, 1999.
%





\bibitem{Braun:2013cb} 
  A.~P.~Braun and T.~Watari,
  ``On Singular Fibres in F-Theory,''
  JHEP {\bf 1307}, 031 (2013)
  doi:10.1007/JHEP07(2013)031
  [arXiv:1301.5814 [hep-th]].


\bibitem{Boubaki.Lie46}
N.~Boubaki, {\it Lie Groups and Lie Algebras}, Chapters 4-6, Springer-Verlag Berlin Heidelberg, 2002.
%





\bibitem{Bershadsky:1996nh}
M.~Bershadsky, K.~A. Intriligator, S.~Kachru, D.~R. Morrison, V.~Sadov, and
  C.~Vafa.
\newblock {Geometric singularities and enhanced gauge symmetries}.
\newblock {\em Nucl. Phys.}, B481:215--252, 1996.
%

\bibitem{Bonetti:2013ela} 
  F.~Bonetti, T.~W.~Grimm and S.~Hohenegger,
  ``One-loop Chern-Simons terms in five dimensions,''
  JHEP {\bf 1307}, 043 (2013)
  doi:10.1007/JHEP07(2013)043
  [arXiv:1302.2918 [hep-th]].
  
  
  \bibitem{Bonora:2010bu} 
  L.~Bonora and R.~Savelli,
  ``Non-simply-laced Lie algebras via F theory strings,''
  JHEP {\bf 1011}, 025 (2010)
  doi:10.1007/JHEP11(2010)025
  [arXiv:1007.4668 [hep-th]].

\bibitem{Braun:2014oya} 
  V.~Braun and D.~R.~Morrison,
  ``F-theory on Genus-One Fibrations,''
  JHEP {\bf 1408}, 132 (2014)
  doi:10.1007/JHEP08(2014)132
  [arXiv:1401.7844 [hep-th]].
\bibitem{Bernadoni}
Bernardoni, Fabio; Cacciatori, Sergio L.; Cerchiai, Bianca L.; Scotti, Antonio. 
``Mapping the geometry of the F$_4$ group," 
 Adv. Theor. Math. Phys. {\bf 12} (2008), no. 4, 889--994.
\bibitem{Cadavid:1995bk} 
  A.~C.~Cadavid, A.~Ceresole, R.~D'Auria and S.~Ferrara,
  ``Eleven-dimensional supergravity compactified on Calabi-Yau threefolds,''
  Phys.\ Lett.\ B {\bf 357}, 76 (1995)
  doi:10.1016/0370-2693(95)00891-N
  [hep-th/9506144].


\bibitem{Carter}
R.~W. Carter.
\newblock {\em Lie algebras of finite and affine type}, volume~96 of {\em
  Cambridge Studies in Advanced Mathematics}.
\newblock Cambridge University Press, Cambridge, 2005.
%

\bibitem{Cattaneo:2013vda} 
  A.~Cattaneo,
  ``Crepant resolutions of Weierstrass threefolds and non-Kodaira fibres,''
  arXiv:1307.7997 [math.AG].

\bibitem{CDE}
A.~Collinucci, F.~Denef, and M.~Esole.
\newblock {D-brane Deconstructions in IIB Orientifolds}.
\newblock {\em JHEP}, 02:005, 2009.



\bibitem{Cvetic:2012xn} 
  M.~Cvetic, T.~W.~Grimm and D.~Klevers,
  ``Anomaly Cancellation And Abelian Gauge Symmetries In F-theory,''
  JHEP {\bf 1302}, 101 (2013)
  doi:10.1007/JHEP02(2013)101
  [arXiv:1210.6034 [hep-th]].

\bibitem{deBoer:2001wca}
J.~de~Boer, R.~Dijkgraaf, K.~Hori, A.~Keurentjes, J.~Morgan, D.~R. Morrison,
  and S.~Sethi.
\newblock {Triples, fluxes, and strings}.
\newblock {\em Adv. Theor. Math. Phys.}, 4:995--1186, 2002.
%
\bibitem{Formulaire}
P.~Deligne.
\newblock Courbes elliptiques: formulaire d'apr{\`e}s {J}. {T}ate.
\newblock In {\em Modular functions of one variable, {IV} ({P}roc. {I}nternat.
  {S}ummer {S}chool, {U}niv. {A}ntwerp, {A}ntwerp, 1972)}, pages 53--73.
  Lecture Notes in Math., Vol. 476. Springer, Berlin, 1975.
%
\bibitem{DelZotto:2017pti} 
  M.~Del Zotto, J.~J.~Heckman and D.~R.~Morrison,
  ``6D SCFTs and Phases of 5D Theories,''
  arXiv:1703.02981 [hep-th].

\bibitem{Diaconescu:1998cn} 
  D.~E.~Diaconescu and R.~Entin,
  ``Calabi-Yau spaces and five-dimensional field theories with exceptional gauge symmetry,''
  Nucl.\ Phys.\ B {\bf 538}, 451 (1999)
  doi:10.1016/S0550-3213(98)00689-0
  [hep-th/9807170].
  

%
\bibitem{EFY}
M.~Esole, J.~Fullwood, and S.-T. Yau.
\newblock {$D_5$ elliptic fibrations: non-Kodaira fibers and new orientifold
  limits of F-theory}.
\newblock Commun.\ Num.\ Theor.\ Phys.\  {\bf 09}, no. 3, 583 (2015)
 \newblock  doi:10.4310/CNTP.2015.v9.n3.a4
\newblock [arXiv:1110.6177 [hep-th]].
%
\bibitem{EJJN1}
M.~Esole, S.~G. Jackson, R.~Jagadeesan, and A.~G. No{\"e}l.
\newblock {Incidence Geometry in a Weyl Chamber I: GL$_n$}, 
\newblock arXiv:1508.03038 [math.RT].
\bibitem{EJJN2}
M.~Esole, S.~G. Jackson, R.~Jagadeesan, and A.~G. No{\"e}l.
\newblock {Incidence Geometry in a Weyl Chamber II: SL$_n$}.
\newblock 2015.   arXiv:1601.05070 [math.RT].
%
%
%
%
\bibitem{MMR.I}
M.~Esole, R.~Jagadeesan, and M.~J. Kang.
\newblock To appear.
%
\bibitem{MP}
M.~Esole and P.~Jefferson.
\newblock To appear.

\bibitem{MP}
M.~Esole and S.~Pasterski.
\newblock To appear.
\bibitem{MMP1} 
  M.~Esole, P.~Jefferson and M.~J.~Kang,
  ``Euler Characteristics of Crepant Resolutions of Weierstrass Models,''
  arXiv:1703.00905 [math.AG].
%
\bibitem{EKY}
M.~Esole, M.~J. Kang, and S.-T. Yau.
\newblock {A New Model for Elliptic Fibrations with a Rank One Mordell-Weil
  Group: I. Singular Fibers and Semi-Stable Degenerations}.
\newblock 2014.
%

\bibitem{Esole:2012tf} 
  M.~Esole and R.~Savelli,
  ``Tate Form and Weak Coupling Limits in F-theory,''
  JHEP {\bf 1306}, 027 (2013)
  doi:10.1007/JHEP06(2013)027
  [arXiv:1209.1633 [hep-th]].

\bibitem{Esole:2015xfa} 
  M.~Esole and S.~H.~Shao,
  ``M-theory on Elliptic Calabi-Yau Threefolds and 6d Anomalies,''
  arXiv:1504.01387 [hep-th].


\bibitem{ESY1}
M.~Esole, S.-H. Shao, and S.-T. Yau.
\newblock {Singularities and Gauge Theory Phases}.
\newblock {\em Adv. Theor. Math. Phys.}, 19:1183--1247, 2015.

\bibitem{ESY2}
M.~Esole, S.-H. Shao, and S.-T. Yau.
\newblock {Singularities and Gauge Theory Phases II}.
\newblock {\em Adv. Theor. Math. Phys.}, 20:683--749, 2016.

\bibitem{EY}
M.~Esole and S.-T. Yau.
\newblock {Small resolutions of SU(5)-models in F-theory}.
\newblock {\em Adv. Theor. Math. Phys.}, 17:1195--1253, 2013.


\bibitem{Figueroa}
Figueroa-O?Farrill, 
``A Geometric Construction of the Exceptional Lie Algebras F$_4$ and E$_8$,"
J. Commun. Math. Phys. (2008) 283: 663. doi:10.1007/s00220-008-0581-7

\bibitem{Fullwood:SVW}
J.~Fullwood.
\newblock {On generalized Sethi-Vafa-Witten formulas}.
\newblock {\em J. Math. Phys.}, 52:082304, 2011.
%
\bibitem{FH2}
J.~Fullwood and M.~van Hoeij.
\newblock {On stringy invariants of GUT vacua}.
\newblock {\em Commun. Num. Theor Phys.}, 07:551--579, 2013.
%





\bibitem{Fullwood:2012kj} 
  J.~Fullwood and M.~van Hoeij,
  ``On stringy invariants of GUT vacua,''
  Commun.\ Num.\ Theor Phys.\  {\bf 07}, 551 (2013)
  doi:10.4310/CNTP.2013.v7.n4.a1
  [arXiv:1211.6077 [math.AG]].
\bibitem{Fulton.Intersection}
W.~Fulton.
\newblock {\em Intersection theory}, Springer-Verlag, Berlin, second edition, 1998.
%

\bibitem{GM1}
A.~Grassi and D.~R. Morrison.
\newblock Group representations and the Euler characteristic of elliptically
  fibered Calabi-Yau threefolds.
\newblock {\em J. Algebraic Geom.}, 12(2):321--356, 2003.


\bibitem{GM2} 
  A.~Grassi and D.~R.~Morrison,
  ``Anomalies and the Euler characteristic of elliptic Calabi-Yau threefolds,''
  Commun.\ Num.\ Theor.\ Phys.\  {\bf 6}, 51 (2012)
  doi:10.4310/CNTP.2012.v6.n1.a2
%

\bibitem{Grimm:2015zea} 
  T.~W.~Grimm and A.~Kapfer,
  ``Anomaly Cancelation in Field Theory and F-theory on a Circle,''
  JHEP {\bf 1605}, 102 (2016)
  doi:10.1007/JHEP05(2016)102
  [arXiv:1502.05398 [hep-th]].



\bibitem{Haghighat:2014vxa} 
  B.~Haghighat, A.~Klemm, G.~Lockhart and C.~Vafa,
  ``Strings of Minimal 6d SCFTs,''
  Fortsch.\ Phys.\  {\bf 63}, 294 (2015)
  doi:10.1002/prop.201500014
  [arXiv:1412.3152 [hep-th]].
  
  \bibitem{Hayashi:2014kca}
H.~Hayashi, C.~Lawrie, D.~R. Morrison, and S.~Schafer-Nameki.
\newblock {Box Graphs and Singular Fibers}.
\newblock {\em JHEP}, 1405:048, 2014.


\bibitem{Hartshorne}
R.~Hartshorne, {\em Algebraic Geometry},  Graduate Texts in Mathematics 52, Springer-Verlag, 1977. 
\bibitem{IMS}
K.~A. Intriligator, D.~R. Morrison, and N.~Seiberg.
\newblock {Five-dimensional supersymmetric gauge theories and degenerations of
  Calabi-Yau spaces}.
\newblock {\em Nucl.Phys.}, B497:56--100, 1997.
%


\bibitem{Katz:1996xe} 
  S.~H.~Katz and C.~Vafa,
  ``Matter from geometry,''
  Nucl.\ Phys.\ B {\bf 497}, 146 (1997)
  doi:10.1016/S0550-3213(97)00280-0
  [hep-th/9606086].


\bibitem{Katz:2011qp}
S.~Katz, D.~R. Morrison, S.~Schafer-Nameki, and J.~Sully.
\newblock {Tate's algorithm and F-theory}.
\newblock {\em JHEP}, 1108:094, 2011.
%
\bibitem{Kodaira}
K.~Kodaira.
\newblock On compact analytic surfaces. {II}, {III}.
\newblock {\em Ann. of Math. (2) 77 (1963), 563--626; ibid.}, 78:1--40, 1963.
%
%
\bibitem{Kuntzler:2012bu} 
  M.~Kuntzler and S.~Sch\"afer-Nameki,
  ``G-flux and Spectral Divisors,''
  JHEP {\bf 1211}, 025 (2012)
  [arXiv:1205.5688 [hep-th]].


\bibitem{Hayashi:2014kca} 
  H.~Hayashi, C.~Lawrie, D.~R.~Morrison and S.~Schafer-Nameki,
  ``Box Graphs and Singular Fibers,''
  JHEP {\bf 1405}, 048 (2014)
  doi:10.1007/JHEP05(2014)048
  [arXiv:1402.2653 [hep-th]].


\bibitem{Lawrie:2012gg} 
  C.~Lawrie and S.~Sch\"afer-Nameki,
  ``The Tate Form on Steroids: Resolution and Higher Codimension Fibers,''
  JHEP {\bf 1304}, 061 (2013)
  doi:10.1007/JHEP04(2013)061
  [arXiv:1212.2949 [hep-th]].

\bibitem{QLiu.AGAC}
Q.~Liu.
\newblock {\em Algebraic geometry and arithmetic curves}, volume~6 of {\em
  Oxford Graduate Texts in Mathematics}.
\newblock Oxford University Press, Oxford, 2002.
\newblock Translated from the French by Reinie Ern{\'e}, Oxford Science
  Publications.
%
%
\bibitem{Marsano}
J.~Marsano and S.~Schafer-Nameki.
\newblock {Yukawas, G-flux, and Spectral Covers from Resolved Calabi-Yau's}.
\newblock {\em JHEP}, 11:098, 2011.
%
%
\bibitem{Matsuki}
\newblock K.~ Matsuki,  Introduction to the Mori Program. 
\newblock Springer Science \& Business Media, 2013.
\bibitem{Miranda.smooth}
R.~Miranda.
\newblock Smooth models for elliptic threefolds.
\newblock In {\em The birational geometry of degenerations ({C}ambridge,
  {M}ass., 1981)}, volume~29 of {\em Progr. Math.}, pages 85--133.
  Birkh{\"a}user Boston, Mass., 1983.
%
\bibitem{Morrison:2011mb} 
  D.~R.~Morrison and W.~Taylor,
  ``Matter and singularities,''
  JHEP {\bf 1201}, 022 (2012)
  doi:10.1007/JHEP01(2012)022
  [arXiv:1106.3563 [hep-th]].
%
\bibitem{Morrison:2012np} 
  D.~R.~Morrison and W.~Taylor,
  ``Classifying bases for 6D F-theory models,''
  Central Eur.\ J.\ Phys.\  {\bf 10}, 1072 (2012)
  doi:10.2478/s11534-012-0065-4
  [arXiv:1201.1943 [hep-th]].
%
\bibitem{Morrison:1996na}
D.~R. Morrison and C.~Vafa.
\newblock {Compactifications of F theory on Calabi-Yau threefolds. 1}.
\newblock {\em Nucl. Phys.}, B473:74--92, 1996.
%
\bibitem{Morrison:1996pp}
D.~R. Morrison and C.~Vafa.
\newblock {Compactifications of F theory on Calabi-Yau threefolds. 2.}
\newblock {\em Nucl. Phys.}, B476:437--469, 1996.
%

%
%
%
\bibitem{MumfordSuominen}
D.~Mumford and K.~Suominen.
\newblock Introduction to the theory of moduli.
\newblock In {\em Algebraic geometry, {O}slo 1970 ({P}roc. {F}ifth {N}ordic
  {S}ummer-{S}chool in {M}ath.)}, pages 171--222. Wolters-Noordhoff, Groningen,
  1972.



\bibitem{Neron}
A.~N{{\'e}}ron.
\newblock Mod{\`e}les minimaux des vari{\'e}t{\'e}s ab{\'e}liennes sur les
  corps locaux et globaux.
\newblock {\em Inst. Hautes {\'E}tudes Sci. Publ.Math. No.}, 21:128, 1964.



\bibitem{MR1312368}
J.~H. Silverman.
\newblock {\em Advanced topics in the arithmetic of elliptic curves}, volume
  151 of {\em Graduate Texts in Mathematics}.
\newblock Springer-Verlag, New York, 1994.
%
\bibitem{Szydlo.Thesis}
M.~G. Szydlo.
\newblock {\em Flat regular models of elliptic schemes}.
\newblock ProQuest LLC, Ann Arbor, MI, 1999.
\newblock Thesis (Ph.D.)--Harvard University.
%
\bibitem{Tachikawa:2015wka} 
  Y.~Tachikawa,
  ``Frozen singularities in M and F theory,''
  JHEP {\bf 1606}, 128 (2016)
  doi:10.1007/JHEP06(2016)128
  [arXiv:1508.06679 [hep-th]].
%
\bibitem{Tate}
J.~Tate.
\newblock Algorithm for determining the type of a singular fiber in an elliptic
  pencil.
\newblock In {\em Modular functions of one variable, {IV} ({P}roc. {I}nternat.
  {S}ummer {S}chool, {U}niv. {A}ntwerp, {A}ntwerp, 1972)}, pages 33--52.
  Lecture Notes in Math., Vol. 476. Springer, Berlin, 1975.
%

\bibitem{Taylor:2012dr} 
  W.~Taylor,
  ``On the Hodge structure of elliptically fibered Calabi-Yau threefolds,''
  JHEP {\bf 1208}, 032 (2012)
  doi:10.1007/JHEP08(2012)032
  [arXiv:1205.0952 [hep-th]].

\bibitem{Vafa:1996xn}
C.~Vafa.
\newblock {Evidence for F theory}.
\newblock {\em Nucl. Phys.}, B469:403--418, 1996.

\bibitem{Wazir}
R.~Wazir.
\newblock Arithmetic on elliptic threefolds.
\newblock {\em Compositio Mathematica}, 140(03):567--580, 2004.
%
\bibitem{Witten:1996qb} 
  E.~Witten,
  ``Phase transitions in M theory and F theory,''
  Nucl.\ Phys.\ B {\bf 471}, 195 (1996)
  [hep-th/9603150].
\end{thebibliography}
\end{document}